\documentclass[journal=jctcce,manuscript=article]{achemso}

\usepackage[version=3]{mhchem} 
\usepackage[T1]{fontenc}       
\usepackage{lmodern} 

\usepackage{color}
\usepackage{algpseudocode}
\usepackage{algorithm}

\usepackage{mathtools}

\usepackage{graphicx}
\usepackage{hyperref}
\usepackage{subcaption}
 \usepackage{amsmath}
 \usepackage{amssymb}
 \usepackage{amsthm}
\numberwithin{equation}{section}
\numberwithin{figure}{section}
 \usepackage{booktabs}
 \newcommand{\ra}[1]{\renewcommand{\arraystretch}{#1}}

\newcommand{\etal}{\textit{et al}.}

\newtheorem{theorem}{Theorem}
\newtheorem{remark}{Remark}[section]
\DeclareMathOperator*{\argmin}{arg\,min}
\DeclareMathOperator*{\argmax}{arg\,max}

\newcommand\upperBound{\textrm{up}}
\newcommand\lowerBound{\textrm{lo}}

\providecommand{\e}[1]{\ensuremath{10^{#1}}}

\author{Anastasia Kruchinina}
\affiliation{Division of Scientific
  Computing, Department of Information Technology, Uppsala University, Sweden}
\email{anastasia.kruchinina@it.uu.se}
\author{Elias Rudberg}
\email{elias.rudberg@it.uu.se}
\affiliation{Division of Scientific
  Computing, Department of Information Technology, Uppsala University, Sweden}
\author{Emanuel H. Rubensson}
\email{emanuel.rubensson@it.uu.se}
\affiliation{Division of Scientific
  Computing, Department of Information Technology, Uppsala University, Sweden}

\title[Parameterless stopping criteria] {Parameterless stopping
  criteria for \\recursive density matrix expansions}

\begin{document}

\begin{abstract}
Parameterless stopping criteria for recursive polynomial expansions to
construct the density matrix in electronic structure calculations are
proposed. Based on convergence order estimation the new stopping
criteria automatically and accurately detect when the calculation is
dominated by numerical errors and continued iteration does not improve
the result.  Difficulties in selecting a stopping tolerance and
appropriately balancing it in relation to parameters controlling the
numerical accuracy are avoided.  Thus, our parameterless stopping
criteria stand in contrast to the standard approach to stop as soon as
some error measure goes below a user-defined parameter or tolerance.
We demonstrate that the stopping criteria work well both in dense and
sparse matrix calculations and in large-scale self-consistent field
calculations with the quantum chemistry program {\sc Ergo}
(\url{www.ergoscf.org}).

\end{abstract}


\section{Introduction}
An important computational task in electronic structure calculations
based on for example Hartree--Fock \cite{Roothaan} or Kohn--Sham density functional theory \cite{hohen,KohnSham65} is
the computation of the one-electron density matrix $D$ for a given
Fock or Kohn--Sham matrix $F$. The density matrix is the matrix for
orthogonal projection onto the subspace spanned by eigenvectors of $F$
that correspond to occupied electron orbitals:
\begin{align} 
  Fx_i &= \lambda_i x_i, \label{eq:eig_problem} \\
  D \coloneqq &\sum_{i=1}^{n_\textrm{occ}} x_ix_i^T,  \label{eq:D_def}
\end{align}
where the eigenvalues of $F$ are arranged in ascending order
\begin{equation}
 \lambda_1 \leq \lambda_2 \leq  \dots \leq  \lambda_{\textrm{homo}} <
\lambda_{\textrm{lumo}} \leq \dots  \leq \lambda_{N-1} \leq \lambda_{N},
\end{equation}
$n_\textrm{occ}$ is the number of occupied orbitals,
$\lambda_{\textrm{homo}}$ is the highest occupied molecular orbital
(homo) eigenvalue, and $\lambda_{\textrm{lumo}}$ is the lowest
unoccupied molecular orbital (lumo) eigenvalue and where we assume
that there is a gap
\begin{equation}
  \xi \coloneqq \lambda_{\textrm{lumo}} - \lambda_{\textrm{homo}} > 0 
\end{equation}
between eigenvalues corresponding to occupied and unoccupied orbitals.
An essentially direct method to compute $D$ is to compute an
eigendecomposition of $F$ and assemble $D$ according to
\eqref{eq:D_def}. Unfortunately, the computational cost of this
approach increases cubically with system size which limits
applications to rather small systems. Alternative methods have
therefore been developed with the aim to reduce the computational
complexity \cite{Bowler_2012}. One approach is to view the problem as
a matrix function
\begin{equation} \label{eq:stepfun}
  D = \theta(\mu I -F),
\end{equation}
where $\theta$ is the Heaviside function and $\mu$ is located between
$\lambda_{\textrm{homo}}$ and $\lambda_{\textrm{lumo}}$, which makes
\eqref{eq:stepfun} equivalent to the definition in \eqref{eq:D_def} \cite{GoedeckerColombo1994}.
A condition number for the problem of evaluating \eqref{eq:stepfun}
is given by
\begin{equation}
  \kappa \coloneqq \lim_{h\rightarrow 0} \sup_{A:\|A\| = \Delta\epsilon} \frac{\|\theta(\mu I - (F+hA)) - \theta(\mu I - F)\|}{h} = \frac{\Delta\epsilon}{\xi}
\end{equation}
where $\Delta\epsilon$ is the spectral width of
$F$~\cite{book-higham,EHRubensson12,m-accPuri}. We let $\|A\| = \Delta\epsilon$ to make the
condition number invariant both to scaling and shift of the
eigenspectrum of $F$~\cite{EHRubensson12}.

When the homo-lumo gap $\xi > 0$, a function that varies smoothly
between 0 and 1 in the gap can be used in place of \eqref{eq:stepfun}.
To construct such a function, recursive polynomial expansions or
density matrix purification have proven to be particularly simple and
efficient \cite{RudbergRubenssonJPCM11}. 
\begin{algorithm}\caption{Recursive polynomial expansion (general form) \label{alg:rec_exp_general}}
  \begin{algorithmic}[1]
    \State $X_0 = f_0(F)$
    \State $\widetilde{X}_0 = X_0+E_0$
    \While {stopping criterion not fulfilled, for $i = 1,2,\dots$}
    \State $X_{i} = f_i(\widetilde{X}_{i-1})$
    \State $\widetilde{X}_i = X_i+E_i$
    \EndWhile
  \end{algorithmic}
\end{algorithm}
The regularized step function
is built up by the recursive application of low-order polynomials
$f_0,f_1,\dots$, see Algorithm~\ref{alg:rec_exp_general}.  
With this approach, a linear scaling computational cost is achieved
provided that the matrices in the recursive expansion are sufficiently
sparse, which is usually ensured by removing small matrix elements
during the course of the recursive expansion~\cite{benzi_decay}.
In Algorithm~\ref{alg:rec_exp_general} the removal of matrix elements,
also called truncation, is written as an explicit perturbation $E_i$
added to the matrix in each iteration.
Several recursive expansion algorithms fitting into the general form
of Algorithm~\ref{alg:rec_exp_general} have been proposed. Note that here
we are considering methods that operate in orthogonal basis. The
function $f_0$ is usually a first order polynomial that moves all
eigenvalues into the $[0, \ 1]$ interval in reverse order. A natural
choice for the iteration function $f_i, \ i=1,2,\dots$ is the McWeeny
polynomial $3x^2-2x^3$ \cite{RMcWeeny56,PM1998}, which makes
Algorithm~\ref{alg:rec_exp_general} essentially equivalent to the
Newton--Schulz iteration for sign matrix
evaluation~\cite{book-higham}.  Furthermore, algorithms were developed
that do not require beforehand knowledge of $\mu$. Palser and
Manolopoulos proposed a recursive expansion based on the McWeeny
polynomial~\cite{PM1998}. Niklasson proposed a simple and efficient
algorithm based on the second order polynomials $x^2$ and
$2x-x^2$~\cite{Nikl2002}.
We will refer to this algorithm as the SP2 algorithm.  The recursive
application of polynomials gives a rapid increase of the polynomial
order and the computational cost increases only with the logarithm of
the condition number~\cite{Nikl2002,RudbergRubenssonJPCM11}.  The
computational cost can be further reduced by a scale-and-fold
acceleration technique giving an even weaker dependence on the
condition number~\cite{Rub2011}.  Recursive expansion algorithms are
key components in a number of linear scaling electronic structure
codes including {\sc CP2K}~\cite{VandeVondele_2012}, {\sc
  Ergo}~\cite{ergo_web,Ergo2011}, {\sc FreeON}~\cite{FreeON}, {\sc
  Honpas}~\cite{honpas}, and {\sc LATTE}~\cite{LATTE-jcp-2012}. Since
most of the computational work lies in matrix-matrix multiplications,
recursive expansion algorithms are well suited for parallel
implementations~\cite{Borstnik2014,Cawkwell2014,Chow_2015,Weber_2015}
and a competitive alternative to diagonalization also in the dense
matrix case~\cite{Cawkwell2014,Chow_2015}.

Different ways to decide when to stop the iterations have been
suggested. A common approach is to stop when some quantity, measuring
how far the matrix is from idempotency, goes below a predetermined
convergence threshold value. The deviation from idempotency has been
measured by the trace \cite{Daniels_1999,Weber_2015} or some matrix
norm
\cite{Chow_2015,book-higham,Mazziotti,curvy_steps,Suryanarayana2013291,VandeVondele_2012}
of $X_i-X_i^2$ sometimes scaled by for example the matrix dimension.
%
%
However, since the recursive expansion is at least quadratically
convergent, what one usually wants is to continue iterating until the
idempotency error does not anymore decrease substantially. This
happens when any further substantial decrease is prevented by rounding
errors or errors due to removal of matrix elements. 

To find a proper relation between matrix element removal and the
parameter measuring idempotency can be a delicate task, often left to
the user of the routine.
However, a few attempts to automatically detect when numerical errors
start to dominate exist in the literature. Palser and Manolopoulos
noted that with their expansions, the functional $\textrm{Tr}[X_iF]$
decreases monotonically in exact arithmetics and suggested to stop on
its first increase which should be an indication of
stagnation~\cite{PM1998}.  A similar criterion for the SP2 expansion
was proposed by Cawkwell \etal~\cite{Cawkwell2014}. In this case, the
iterations are stopped on an increase of the idempotency error measure
$|\textrm{Tr}[X_i-X_i^2]|$. However, the value of the functional or
the idempotency error may continue to decrease without significant
improvement of the accuracy. In such cases, the computational effort
in last iterations is no longer justified.
In the present work, we propose new parameterless stopping criteria
based on convergence order estimation. The stopping criteria are
general and can be used both in the dense and sparse matrix cases
using different strategies for truncation, and with different choices
of polynomials.


\section{Parameterless stopping criteria}\label{sec:paramless}

The iterations of density matrix expansions can be divided into three
phases \cite{note_phases,interior_eigenvalues_2014}: 1) the
conditioning phase where the deviation from idempotency decreases less
than quadratically or not at all, 2) the purification phase where the
idempotency error decreases at least quadratically, and 3) the
stagnation phase where the idempotency error again does not decrease
significantly or at all,
see Figure~\ref{fig:phases}.

\begin{figure}
  \begin{center}
    \includegraphics[width=0.9\textwidth]{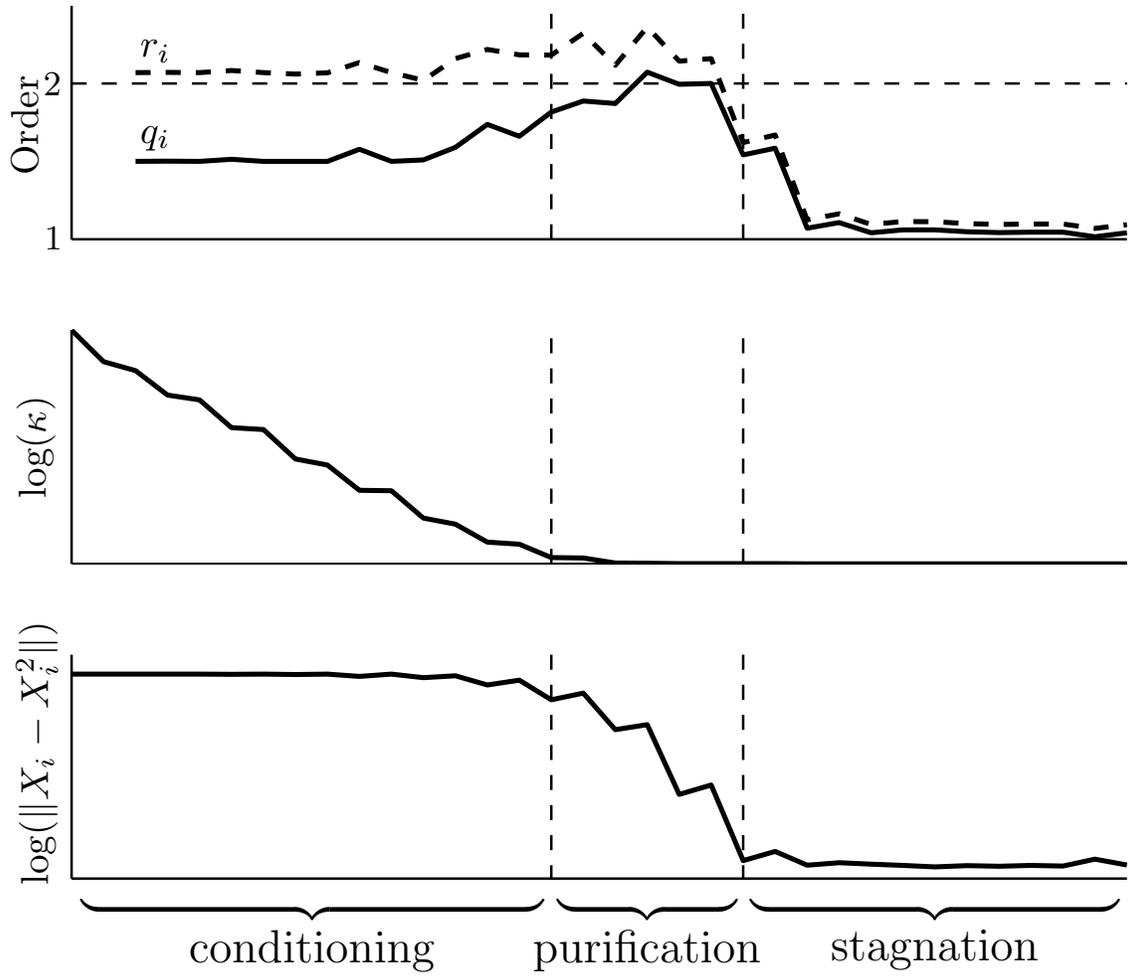}
  \end{center}
\caption{Illustration of the three phases for a recursive expansion of
  order $q=2$ based on the polynomials $x^2$ and $2x-x^2$ (see Section~\ref{sec:sp2}).  In the
  conditioning phase the matrix does not come closer to idempotency
  but the condition number, $\kappa$, is lowered.  In the purification
  phase the condition number is close to 1 and idempotency is
  approached quadratically. In the stagnation phase numerical errors
  start to dominate and the matrix again does not come closer to
  idempotency. The upper panel shows what we call the observed orders
  of convergence $q_i$ and $r_i$.  Throughout the conditioning and
  purification phases $r_i \geq 2$ but in the stagnation phase $r_i <
  2$.
  \label{fig:phases}}
\end{figure}
Here, we propose new parameterless stopping criteria designed to
automatically and accurately detect the transition between
purification and stagnation, so that the procedure can be stopped
without superfluous iterations in the stagnation phase.
We measure the deviation from idempotency by 
\begin{equation} \label{eq:idem_error_spectral}
  e_i \coloneqq \|X_i-X_i^2\|_2.
\end{equation}
We recall that an iterative method has asymptotic order of convergence $q$ if it
in exact arithmetics generates a sequence of errors $e_1,e_2,\dots$
such that
\begin{equation}
\lim_{i\rightarrow \infty} \frac{e_{i}}{e_{i-1}^q} = C^\infty,
\end{equation}
where $C^\infty$ is an asymptotic constant. The order of convergence
can also be observed numerically by 
\begin{equation} \label{eq:qi_def}
q \approx
\frac{\log(e_i/C^\infty)}{\log(e_{i-1})}  \eqqcolon q_i.
\end{equation}
Our stopping criteria are based on the detection of a discrepancy
between the asymptotic and observed orders of convergence.  When the
stagnation phase is entered numerical errors start to dominate,
leading to a fall in the observed order of convergence, see the upper
panel in Figure~\ref{fig:phases}. 

Since the observed order can be significantly smaller than the
asymptotic order also in the initial conditioning phase, an issue is
how to determine when the purification phase has started and one can
start to look for a drop in the order $q_i$.  A similar problem of
determining the iteration when to start to check the stopping
criterion appears in the method described by Cawkwell
\etal~\cite{Cawkwell2014}. Our solution is to replace the asymptotic
constant $C^\infty$ in~\eqref{eq:qi_def} with a larger value such that
the observed order of convergence \emph{in exact arithmetics} is
always larger than or equal to the asymptotic order of convergence.
In other words we want to find $C_q$, as small as possible, such that
in exact arithmetics
\begin{equation}\label{eq:Cq_relation}
 r_i \coloneqq \frac{\log(e_i/C_q)}{\log(e_{i-1})} \geq q  
\end{equation}
for all $i$.
One may let $C_q$ vary over the iterations but we will later see that
it is usually sufficient to use a single value $C_q$ for the whole
expansion. Important is that, in exact arithmetics, $r_i \geq q$ for all $i$.  
In the presence of numerical errors $r_i$ is
significantly smaller than $q$ only in the stagnation phase. We may therefore
start to look for a drop in $r_i$ immediately.  As soon as the
observed order of convergence, $r_i$, goes significantly below $q$, the procedure
should be stopped, since this indicates the transition between
purification and stagnation.
In this way we avoid the issue of detecting the transition between
conditioning and purification.
See the upper panel in Figure~\ref{fig:phases} for an illustration.

For clarity we note that~\eqref{eq:Cq_relation} is equivalent to
\begin{equation}\label{eq:Cq_relation_2}
  C_q \geq \frac{e_i}{e_{i-1}^q} =
  \frac{\|f_i(X_{i-1})-f_i(X_{i-1})^2\|_2}{\|X_{i-1} -
    X_{i-1}^2\|_2^q}.
\end{equation}
For generality and simplicity we want to assume as little information
as possible about the eigenspectra of $X_i,\ i=0,1,\dots$. We will use
the following theorem to find the smallest possible $C_q$ value
fulfilling~\eqref{eq:Cq_relation_2} with no or few assumptions about
the location of eigenvalues for several recursive expansion
polynomials of interest for appropriate choices of $q$.

\begin{theorem}
\label{Cq_theorem}
  Let $f$ be a continuous function from $[0, \ 1]$ to $[0, \ 1]$
  and assume that the limits
  \begin{equation} \label{eq:limAB}
    \lim_{x \rightarrow a^+} \frac{f(x)-f(x)^2}{(x-x^2)^q}
    \quad
    \text{and} \quad
    \lim_{x \rightarrow (1-a)^-} \frac{f(x)-f(x)^2}{(x-x^2)^q}
  \end{equation}  
  exist for some $q>0$, where $a \in [0, \ 0.5]$.  Let $H$ denote the set of Hermitian
  matrices with all eigenvalues in $[0, \ 1]$ and at least one eigenvalue
  in $[a, \ 1-a]$.  Then,
  \begin{equation}\label{eq:Cq_def}
    \max_{X \in H}\frac{\|f(X)-f(X)^2\|_2}{\|X-X^2\|_2^q} = \max_{x \in [0, 1]} g(x,a),
  \end{equation}
  where 
  \begin{equation}
    g(x,a) \coloneqq \left\{\begin{array}{ll}
    \frac{f(x)-f(x)^2}{(x-x^2)^q} & \textrm{if } a\leq x \leq 1-a, \\
    \frac{f(x)-f(x)^2}{(a-a^2)^q} & \textrm{otherwise,}
    \end{array}\right.
  \end{equation}
  is extended by continuity at $x=0$ and $x=1$ for $a=0$.
\end{theorem}

As suggested by Theorem~\ref{Cq_theorem}, we will choose $C_q
\coloneqq \max_{x \in [0, 1]} g(x,a)$ thereby making sure that
\eqref{eq:Cq_relation_2} is fulfilled. In principle, the value $a$
should be chosen as large as possible to get the smallest possible
$C_q$-value, since a larger $a$ gives a smaller set of matrices $H$
in~\eqref{eq:Cq_def}.  We note that it is possible to let $C_q$ vary
by choosing the largest possible $a$ in every iteration, but in this
work we will attempt to use a single value for the whole expansion
whenever possible. There is always at least one eigenvalue in the
interval $[0, \ 1]$ in every iteration. Therefore, if, for the given
recursive expansion polynomials $f_i$, the limits \eqref{eq:limAB}
exist with $a=0$ we employ the theorem with $a=0$. Only if the limits
\eqref{eq:limAB} do not exist with $a=0$ will we use the theorem with
$a>0$ and in general get different values of $C_q$ in every
iteration. In such a case some information about the eigenspectrum of
$X_i$ in each iteration $i$ is needed so that $a$ can be chosen
appropriately. The theorem should be invoked with $q$ equal to the
order of convergence of the recursive expansion.

\begin{proof}[Proof of Theorem~\ref{Cq_theorem}]
  Continuity of $f(x)$ together with existence of the two limits
  \eqref{eq:limAB} (needed in case $a=0$) implies
  existence of $\max_{x \in [0, 1]} g(x,a)$.
   Let $X$ be a matrix in $H$. Then, 
   \allowdisplaybreaks
  \begin{align}
    \frac{\|f(X)-f(X)^2\|_2}{\|X-X^2\|_2^q} \leq &    
    \max_{x\in[0,1]} 
    \max_{\begin{subarray}{c}y\in[a,1-a] \\ |y-0.5|\leq|x-0.5| \end{subarray}}
    \frac{f(x)-f(x)^2}{(y-y^2)^q} \label{eq:p1} \\
    =  & 
    \max\left(
    \max_{x\in[a,1-a]} \max_{|y-0.5|\leq|x-0.5|} \frac{f(x)-f(x)^2}{(y-y^2)^q}, 
    \right. 
    \\
    & \phantom{\max\left(\right.} 
    \left.\max_{x\in[0,1]\setminus[a,1-a]} \max_{y\in[a,1-a]} \frac{f(x)-f(x)^2}{(y-y^2)^q}
    \right) \notag 
    \\
    = & 
    \max\left(
    \max_{x\in[a,1-a]} \frac{f(x)-f(x)^2}{(x-x^2)^q},
    \right. 
    \\
    & \phantom{\max\left(\right.} 
    \left.\max_{x\in[0,1]\setminus[a,1-a]} \frac{f(x)-f(x)^2}{(a-a^2)^q}
    \right) \notag \\
    = & \max_{x\in[0,1]} g(x,a).  
  \end{align}
  The maximum in \eqref{eq:Cq_def} is attained for any matrix in $H$
  with an eigenvalue $\lambda = \argmax_{x\in[0,1]} g(x,a)$.
\end{proof}
\begin{remark}
Inequality \eqref{eq:p1} can be interpreted as follows. The arguments
of the maxima, $x$ and $y$, play the roles of the eigenvalues of $X$
that define $\|f(X)-f(X)^2\|_2$ and $\|X-X^2\|_2$, respectively.
While, for all we know, $x$ can be any eigenvalue, $y$ is necessarily
the eigenvalue of $X$ that is closest to 0.5. In other words,
\begin{equation}
  y = \argmin_{\lambda\in\{\lambda_i\}} |\lambda-0.5|
\end{equation}
and
\begin{equation}
  \|X-X^2\|_2 = \max_{\lambda\in\{\lambda_i\}} \lambda-\lambda^2 =
  y-y^2,
\end{equation}
where  $\{\lambda_i\}$ are the eigenvalues of $X$.
Since there is at least one eigenvalue at $x$ and at least one
eigenvalue in $[a, \ 1-a]$ the constraints on $y$ in \eqref{eq:p1}
follow.
\end{remark}

By construction, the observed convergence order $r_i$ may in exact
arithmetics end up exactly equal to $q$. We want to detect when
numerical errors start to dominate and $r_i$ drops significantly below
$q$ but at the same time we want to disregard small perturbations that
cause $r_i$ to go only slightly below $q$. In practice we will
therefore use a parameter $\tilde{q} < q$ in place of $q$. In other
words, our stopping criteria will be on the form \emph{stop as soon as
  $\log(e_i/C_q)/\log(e_{i-1}) < \tilde{q}$}.  In this work we will
mainly consider second order expansions with $q=2$ and will then use
$\tilde{q} = 1.8$.

The spectral norm is often expensive to compute and one may therefore
want to use an estimate in its place. One cheap alternative is the
Frobenius norm
\begin{equation} \label{eq:frobnorm}
  v_i \coloneqq \|X_i-X_i^2\|_F = \sqrt{ \sum_j \left(\lambda_j - \lambda_j^2\right)^2},
\end{equation}
where $\{\lambda_j\}$ are the eigenvalues of $X_i$. 
We expect the Frobenius norm to be a good estimate to the spectral
norm when we are close to convergence since then many eigenvalues are
clustered around 0 and 1 and do not contribute significantly to the
sum in~\eqref{eq:frobnorm}.
Since $\|X_i-X_i^2\|_2 \leq \|X_i-X_i^2\|_F$ we have that
\begin{equation} \label{eq:K_def}
  v_i = K_i e_i
\end{equation}
for some $K_i \geq 1$. 
Assume that for a given $q$
\begin{equation} \label{eq:K_ratio}
  K_i \leq K_{i-1}^q.
\end{equation}
Then,
\begin{equation}
  v_i = K_i e_i \leq K_{i-1}^q C_q e_{i-1}^q = C_q v_{i-1}^q,
\end{equation}
and, assuming that $v_{i-1}<1$, 
\begin{equation}
  \frac{\log(v_i/C_q)}{\log(v_{i-1})} \geq q,
  \label{eq:estim_order_using_frob_norm}
\end{equation}
which means that if the assumptions above hold we will not stop
prematurely.
In the following sections we will see under what
conditions~\eqref{eq:K_ratio} holds for specific choices of
polynomials for the recursive expansion.

For very large systems, it may happen that $v_{i}$ never goes below 1
because of the large number of eigenvalues that contribute to the sum
in \eqref{eq:frobnorm}, in such a case the stopping criterion in the
suggested above form will not be checked. One may then make use of the
so-called mixed norm~\cite{mixedNormTrunc}. Let the
matrix be divided into square submatrices of equal size, padding the
matrix with zeros if needed. One can define the mixed norm as the
spectral norm of a matrix with elements equal to the Frobenius norms
of the obtained submatrices.  It can be shown that
\begin{align}
 e_i \leq m_i \leq v_i
\end{align}
where $m_i \coloneqq \|X_i-X_i^2\|_M$ is the mixed norm of $X_i-X_i^2$
in iteration $i$. The result for the Frobenius norm above that we will
not stop prematurely therefore holds also for the mixed norm under the
corresponding assumptions, i.e.~\eqref{eq:K_def}
and~\eqref{eq:K_ratio} with $v_i$ replaced by $m_i$ and $m_{i-1} <
1$. However, with a fixed submatrix block size the asymptotic behavior
of the mixed norm follows that of the spectral norm.  Thus, the
quality of the stopping criterion will not deteriorate with increasing
system size. One may therefore consider using the mixed norm for large
systems.

We will in the following mainly focus on the regular and accelerated
McWeeny and second order spectral projection polynomials.


\section{McWeeny polynomial}
We will first consider a recursive polynomial expansion based on the
McWeeny polynomial $p_{\mathrm{mw}}(x) \coloneqq 3x^2 -
2x^3$~\cite{RMcWeeny56,PM1998}, i.e.
Algorithm~\ref{alg:rec_exp_general} with
\begin{equation}
\left\{
\begin{aligned}
f_0(x) & = \displaystyle\frac{\mu-x}{2\max(\lambda_{\textrm{max}}-\mu, \mu-\lambda_{\textrm{min}})}+0.5, \quad \\
f_i(x) & = p_{\textrm{mw}}(x), \quad i = 1,2,\dots.
\end{aligned}
\right.
\end{equation}
Here, $\lambda_\textrm{min}$ and $\lambda_\textrm{max}$ are the
extremal eigenvalues of $F$ or bounds thereof,
i.e. 
\begin{equation}
\lambda_\textrm{min} \leq \lambda_1 \textrm{ and } 
\lambda_\textrm{max} \geq \lambda_N.
\end{equation}
We want to find the smallest value
$C_2^{\textrm{mw}}$ such that
\begin{equation}
  \frac{\log(e_i/C_2^{\textrm{mw}})}{\log(e_{i-1})} \geq 2
\end{equation}
in exact arithmetics. Note that 
\begin{equation}
  \lim_{x \rightarrow 0^+} \frac{p_{\textrm{mw}}(x)-p_{\textrm{mw}}(x)^2}{(x-x^2)^2} = 
  \lim_{x \rightarrow 1^-} \frac{p_{\textrm{mw}}(x)-p_{\textrm{mw}}(x)^2}{(x-x^2)^2} = 3.
\end{equation} 
Thus, we may invoke Theorem~\ref{Cq_theorem} with $f =
p_{\textrm{mw}}$, $a=0$ and $q=2$ which gives
\begin{align}
  C_2^{\textrm{mw}} 
  = \max_{x\in[0,1]} \frac{p_{\textrm{mw}}(x)-p_{\textrm{mw}}(x)^2}{(x-x^2)^2} 
  = \max_{x\in[0,1]} (3 + 4 x - 4 x^2) = 4
\end{align}
and we suggest to stop the expansion as soon as
\begin{equation}
  \frac{\log(e_i/4)}{\log(e_{i-1})} < 1.8.
\end{equation}

We note that the McWeeny polynomial can be defined as the polynomial
with fixed points at 0 and 1 and one vanishing derivative at 0 and
1. For polynomials with fixed points at 0 and 1 and $q-1$ vanishing
derivatives at 0 and 1, sometimes referred to as the Holas polynomials
\cite{Holas_2001}, it can be shown that the smallest value
$C_q^{\textrm{h}}$ such that
\begin{equation}
  \frac{\log(e_i/C_q^{\textrm{h}})}{\log(e_{i-1})} \geq q
\end{equation}
is $C_q^{\textrm{h}} = 4^{q-1}$.


\subsection{Acceleration}
Prior knowledge of the homo and lumo eigenvalues makes it possible to
use a scale-and-fold technique to accelerate
convergence~\cite{Rub2011}. In each iteration the eigenspectrum is
stretched out so that the projection polynomial folds the
eigenspectrum over itself.  This technique is quite general and may be
applied to a number of recursive expansions making use of various
polynomials. In case of the McWeeny polynomial the eigenspectrum is
stretched out around 0.5, see Algorithm~\ref{alg:mcw-acc}.  The amount
of stretching is determined by the parameter $\tilde{\xi}$ which is an
estimate of the homo-lumo gap $\xi$ such that at least one eigenvalue
of $F$ is in $[\mu-\tilde{\xi}/2, \ \mu+\tilde{\xi}/2]$.  With
$\alpha=1$ on line~\ref{alg:mcw-acc_alpha}, Algorithm~\ref{alg:mcw-acc} reduces to the
regular McWeeny expansion discussed above.


\begin{algorithm}\caption{McW-ACC algorithm\label{alg:mcw-acc}}
  \begin{algorithmic}[1]
    \State \textbf{input:} $F$, $\lambda_\textrm{min}$,
    $\lambda_\textrm{max}$, $\mu$, $\tilde{\xi}$
    \State $\gamma = 2\max(\lambda_{\textrm{max}}-\mu, \mu-\lambda_{\textrm{min}})$
    \State $\beta_0 = 0.5(1-\tilde{\xi}/\gamma)$
    \State $X_0 = \frac{\mu I-F}{\gamma}+0.5I$
    \State $\widetilde{X}_0 = X_0+E_0$
    \State $e_0 = \|\widetilde{X}_0-\widetilde{X}_0^2\|_2$
    \For {$i = 1,2,\dots$}

    \State $\alpha = 3/(12\beta_{i-1}^2-18\beta_{i-1}+9)^{1/2}$  \label{alg:mcw-acc_alpha}
    \State $X_s = \alpha(\widetilde{X}_{i-1}-0.5I)+0.5I$
    \State $X_i = 3X_s^2-2X_s^3$    \label{alg:mcw-acc_x_i}
    \State $\beta_s = \alpha(\beta_{i-1}-0.5)+0.5$
    \State $\beta_i = 3\beta_s^2-2\beta_s^3$
    

    \State $\widetilde{X}_i = X_i+E_i$
    \State $e_i = \|\widetilde{X}_i-\widetilde{X}_i^2\|_2$
    \If{$\frac{\log(e_i/4)}{\log(e_{i-1})} < 1.8$}
    \State $n=i$
    \State \textbf{break}
    \EndIf
    \EndFor
    \State \textbf{output:} $n$, $\widetilde{X}_n$
  \end{algorithmic}
\end{algorithm}

As for the regular McWeeny expansion we want to find the smallest
value $C_2^{\textrm{mwa}}$ such that
\begin{equation}
  \frac{\log(e_i/C_2^{\textrm{mwa}})}{\log(e_{i-1})} \geq 2. 
\end{equation}
For the sake of analysis we introduce  
\begin{align}
  p_{\textrm{mwa}}(x, \beta) &\coloneqq p_{\textrm{mw}}\left(3(12\beta^2-18\beta+9)^{-1/2}(x-0.5) + 0.5\right).
\end{align}
Note that line~\ref{alg:mcw-acc_x_i} in Algorithm~\ref{alg:mcw-acc} is equivalent to
$X_i = p_{\textrm{mwa}}(\widetilde{X}_{i-1}, \beta_{i-1})$.
Furthermore, let
\begin{equation}
  g(x,\beta) = 
  \begin{dcases}
    g_1(x,\beta) & \textrm{if } x\in [\beta,\ 1-\beta], \\
    g_2(x,\beta) & \textrm{otherwise,}
  \end{dcases}
\end{equation}
where 
\begin{align}
 g_1(x, \beta) &= \frac{p_{\textrm{mwa}}(x, \beta)-p_{\textrm{mwa}}(x, \beta)^2}{(x-x^2)^2},\\
 g_2(x, \beta) &= \frac{p_{\textrm{mwa}}(x, \beta)-p_{\textrm{mwa}}(x, \beta)^2}{(\beta-\beta^2)^2}, 
\end{align} 
 which are plotted in Figure~\ref{fig:mwacc_gx} for $\beta = 0.3$.

 \begin{figure}[ht!]
        \centering
        \captionsetup[subfigure]{justification=centering}
        \begin{subfigure}[t]{.45\textwidth}
        \centering
	  \includegraphics[width=\textwidth]{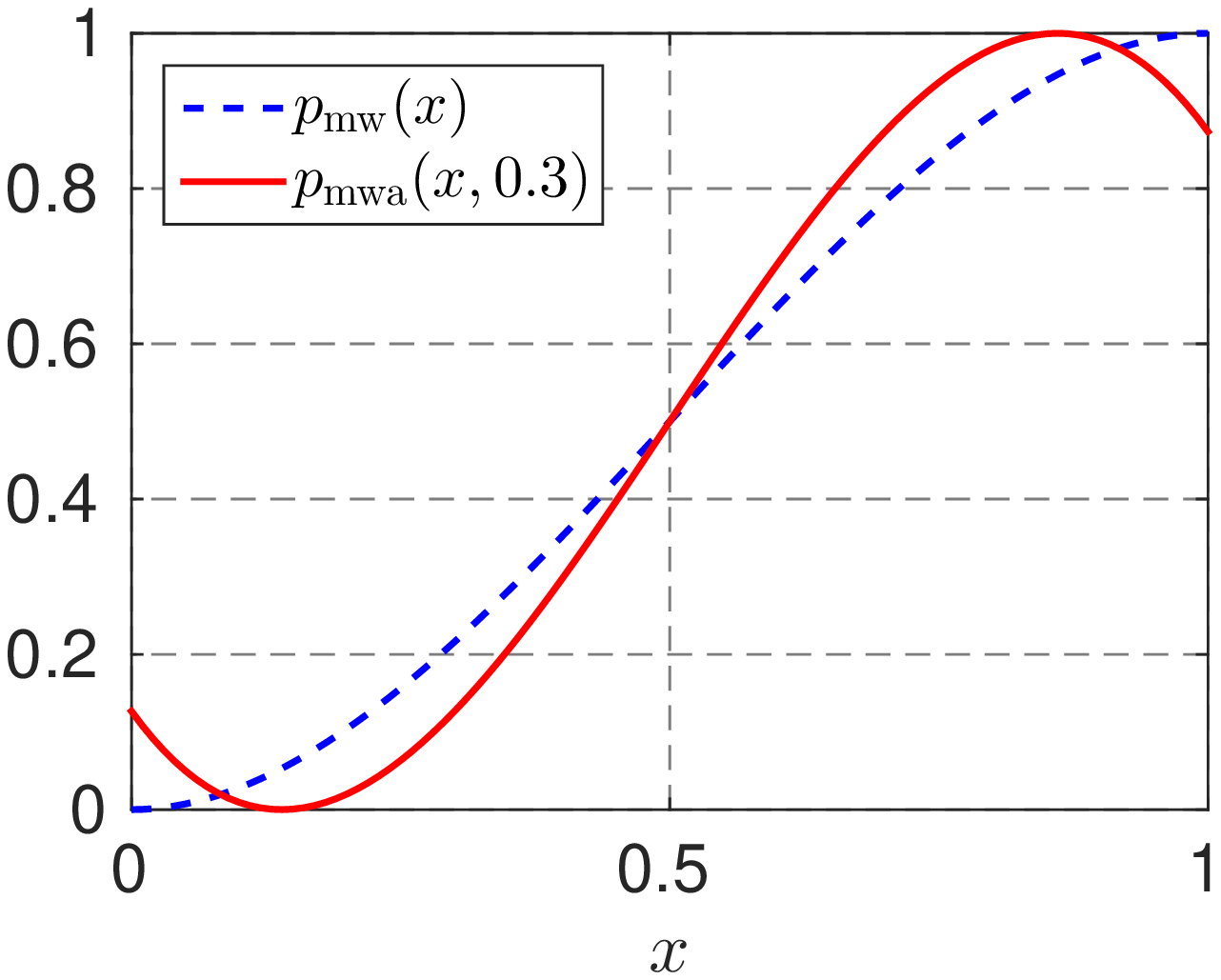}
	\end{subfigure}
	\begin{subfigure}[t]{.45\textwidth}
	\centering
	  \includegraphics[width=\textwidth]{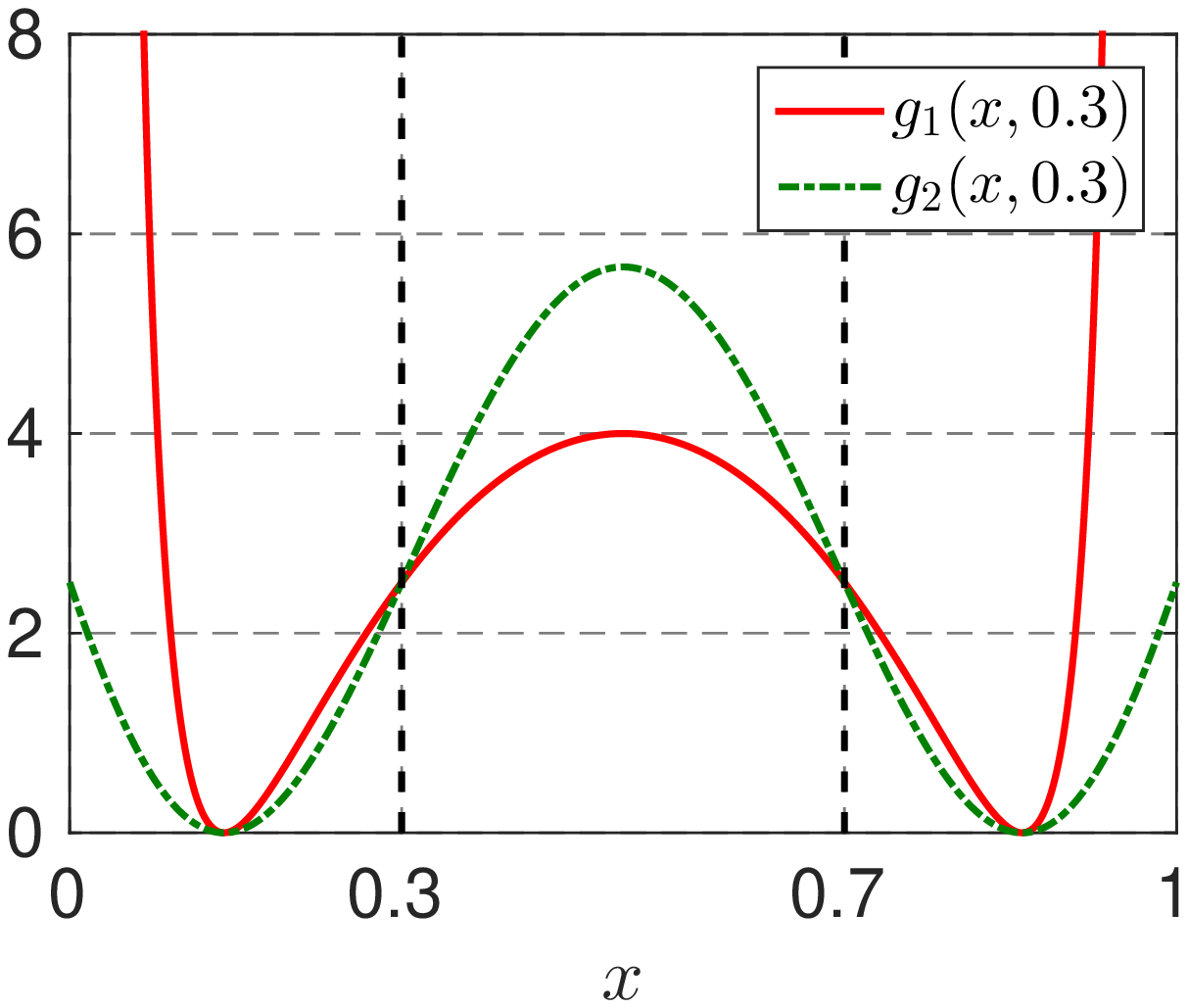}
	\end{subfigure}
  \caption{The left panel shows the regular $p_{\textrm{mw}}(x)$ and accelerated $p_{\textrm{mwa}}(x, \beta)$ McWeeny polynomials with $\beta = 0.3$. The right panel shows the functions $g_1(x, \beta)$ and $g_2(x, \beta)$ with $\beta = 0.3$. Black vertical dashed lines indicate $\beta$ and $1-\beta$. }
  \label{fig:mwacc_gx}
\end{figure}
 
Note that when $\beta>0$, the limits $\lim_{x\to 0^+} g_1(x, \beta)$
and $\lim_{x\to 1^-} g_1(x, \beta)$ do not exist. However, we may use
Theorem~\ref{Cq_theorem} with $a>0$.
We have that 
\begin{align}
  &\lim_{x \rightarrow \beta^+} g_1(x, \beta) =  \lim_{x \rightarrow (1-\beta)^-} g_1(x, \beta)  = 
\begin{dcases}
\frac{16 (\beta-1.5)^2 (\beta-0.75)^2}{(\beta-1)^2 \left(4\beta^2-6\beta+3\right)^3} & \text{ if } \beta > 0,\\
3 & \text{ if } \beta = 0.
\end{dcases}
\end{align} 
Note that 
\begin{align}
 \lim_{\beta \rightarrow 0^+}\lim_{x \rightarrow \beta^+} g_1(x, \beta)  \neq \lim_{x \rightarrow 0^+} g_1(x, 0) 
\end{align} 
due to discontinuity of $g_1(x, \beta)$ at $x=0$ for $\beta > 0$.
Since there is at least one eigenvalue in $[\mu-\tilde{\xi}/2,\
  \mu+\tilde{\xi}/2]$, at least one eigenvalue of $X_i$ will be in the
interval $[\beta_i, \ 1-\beta_i]$ in each iteration $i$.  Therefore, we
may for each iteration $i$ invoke Theorem~\ref{Cq_theorem} with $f(x)
= p_{\textrm{mwa}}(x, \beta_i)$, $a = \beta_i$ and $q = 2$.  We have
that
\begin{align}
\max_{x\in[0,1]}g(x, \beta_i) &= \max{ \left( \max_{x\in[\beta_i, 1-\beta_i]}g_1(x, \beta_i), \max_{x \in [0,1] \setminus [\beta_i, 1-\beta_i]} g_2(x, \beta_i) \right) } \\ 
&= g_1(0.5, \beta_i)= 4,
\end{align}
where we used that for $\beta > 0$ the function $g_2(x, \beta)$ is convex on the intervals $[0, \ \beta]$ and $[1-\beta, \ 1]$, and 
\begin{align}
\max_{x \in [0,1] \setminus [\beta, 1-\beta]} g_2(x, \beta) &= g_2(0, \beta) = g_2(1, \beta) = g_2(\beta, \beta) = g_2(1-\beta, \beta) \\
&= g_1(\beta, \beta) \leq \max_{x\in[\beta, 1-\beta]}g_1(x, \beta).
\end{align}
Therefore, $C_2^{\textrm{mwa}} = 4$, which gives the stopping criterion
in Algorithm~\ref{alg:mcw-acc}.

\subsection{Estimation of the order of convergence}
As discussed earlier one may want to use the Frobenius norm in place
of the spectral norm to measure the idempotency error. It is then
desired that \eqref{eq:K_ratio} is fulfilled, at least in iterations
prior to the stagnation phase. Let $X$ be a Hermitian matrix with
eigenvalues $\{\lambda_j:0\leq\lambda_j\leq1\}$ and let
\begin{align}
  y & = \argmin_{\lambda\in\{\lambda_j\}} |\lambda-0.5| > 0, \\
  \beta & = \min(y,1-y). 
\end{align}
Then, 
\begin{align}
\|X-X^2\|_2 & = \beta-\beta^2 \quad \text{and} \\
\|p_{\textrm{mw}}(X)-p_{\textrm{mw}}(X)^2\|_2 & =
p_{\textrm{mw}}(\beta)-p_{\textrm{mw}}(\beta)^2.
\end{align}

Consider the function 
\begin{align}
f(x) &= \frac{x - x^2}{\beta - \beta^2} -  \frac{p_{\textrm{mw}}(x) - p_{\textrm{mw}}(x)^2}{p_{\textrm{mw}}(\beta) - p_{\textrm{mw}}(\beta)^2} \\
&= -\frac{x\, \left(x - 1\right)\, \left(4\, \beta^4 - 8\, \beta^3 + \beta^2 + 3\, \beta - 4\, x^4 + 8\, x^3 - x^2 - 3\, x\right)}{\beta^2\, {\left(\beta - 1\right)}^2\, \left( - 4\, \beta^2 + 4\, \beta + 3\right)}.
\end{align}
The roots of this function on the interval $[0, \ 1]$ are 0, 1, $\beta$, and $1-\beta$, and the function is non-negative on the intervals $[0, \ \beta]$ and $[1-\beta,\ 1]$. Therefore
\begin{align}
\frac{x - x^2}{\beta - \beta^2} \geq \frac{p_{\textrm{mw}}(x) - p_{\textrm{mw}}(x)^2}{p_{\textrm{mw}}(\beta) - p_{\textrm{mw}}(\beta)^2} \geq 0
\end{align}
for $x - x^2 \leq \beta - \beta^2$ and
\begin{align}
&\frac{\|X-X^2\|_F^2 \; \|p_{\textrm{mw}}(X)-p_{\textrm{mw}}(X)^2\|_2^2}{\|X-X^2\|_2^2 \; \|p_{\textrm{mw}}(X)-p_{\textrm{mw}}(X)^2\|_F^2} \\
&=
 \frac{\sum_j \left((\lambda_j - \lambda_j^2) / (\beta -\beta^2)\right)^2 }{\sum_j \left((p_{\textrm{mw}}(\lambda_j) - p_{\textrm{mw}}(\lambda_j)^2) / (p_{\textrm{mw}}(\beta) - p_{\textrm{mw}}(\beta)^2)\right)^2} \geq 1. \nonumber
\end{align}
Thus, we have that $\frac{K_{i-1}^2}{K_i^2} \geq1$ for every $i\geq1$. Since $K_{i-1} \geq 1$  it follows that
\begin{align}
 K_{i-1}^2 \geq K_{i-1} \geq K_{i},
\end{align}
i.e. \eqref{eq:K_ratio} is fulfilled when using the regular McWeeny expansion.

However, in case of the accelerated McWeeny expansion inequality
\eqref{eq:K_ratio} is not always satisfied. Consider for example the
diagonal matrix
\begin{align}
 X = \text{diag}(0, 0, 0, \beta, 1-\beta).
\end{align}
Then,
\begin{align}
p_{\textrm{mwa}}(0, \beta) - p_{\textrm{mwa}}(0, \beta)^2 &= p_{\textrm{mwa}}(\beta, \beta) - p_{\textrm{mwa}}(\beta, \beta)^2
\end{align}
and
\begin{align}
&\frac{\|X-X^2\|_F^2 \; \|p_{\textrm{mwa}}(X)-p_{\textrm{mwa}}(X)^2\|_2}{\|X-X^2\|_2^2 \;  \|p_{\textrm{mwa}}(X)-p_{\textrm{mwa}}(X)^2\|_F} \\
 &= \frac{ 2\left(\beta-\beta^2\right)^2 \left( p_{\textrm{mwa}}(\beta, \beta) - p_{\textrm{mwa}}(\beta, \beta)^2 \right)  }
         {  \left(\beta-\beta^2\right)^2 \sqrt{5\left(p_{\textrm{mwa}}(\beta, \beta) - p_{\textrm{mwa}}(\beta, \beta)^2\right)^2 }}\nonumber\\
 &=\frac{2}{\sqrt{5}} \approx 0.8944. 
\end{align}
Thus, for the accelerated McWeeny expansion it may happen that
$\frac{K_{i-1}^2}{K_i} < 1$ for some $i$.

To summarize, for the regular McWeeny scheme we will not stop
prematurely when the Frobenius norm is used in place of the spectral
norm. In the accelerated scheme we suggest to turn off the
acceleration at the start of the purification phase when its effect
anyway is small. After that the regular scheme is used and the
Frobenius norm may be used. We will discuss how to detect the
transition between the conditioning and purification phases in
Section~\ref{sec:sp2}.


\section{Second order spectral projection (SP2) expansion}\label{sec:sp2}
In the remainder of this article we will focus on recursive polynomial
expansions based on the polynomials $p_{\textrm{sp}}^{(1)}(x)=x^2$ and
$p_{\textrm{sp}}^{(2)}(x) = 2x-x^2$ proposed by
Mazziotti~\cite{Mazziotti2001}.  In an algorithm proposed by Niklasson
the polynomials used in each iteration are chosen based on the trace
of the matrix \cite{Nikl2002}. This algorithm, hereinafter the
original SP2 algorithm, is given by Algorithm~\ref{alg:sp2}, including
the new stopping criterion proposed here.  
\begin{algorithm}\caption{The original SP2 algorithm with new stopping criterion\label{alg:sp2}}
  \begin{algorithmic}[1]
    \State \textbf{input:} $F$, $\lambda_\textrm{min}$, $\lambda_\textrm{max}$
    \State $X_0 = 
    \frac{\lambda_\textrm{max}I - F}
         {\lambda_\textrm{max} - \lambda_\textrm{min}}$
    \State $\widetilde{X}_0 = X_0+E_0$
    \State $e_0 = \|\widetilde{X}_0-\widetilde{X}_0^2\|_2$
    \State $C_2^{\textrm{sp}} =\frac{1}{32} \left(71+17 \sqrt{17}\right)$
    \For{$i = 1,2,\dots$}
    \If{$\textrm{Tr}[\widetilde{X}_{i-1}] > n_\textrm{occ}$} \label{line:cond_sp2}
    \State $X_i = \widetilde{X}_{i-1}^2$
    \State $p_i = 1$
    \Else
    \State $X_i = 2\widetilde{X}_{i-1}-\widetilde{X}_{i-1}^2$
    \State $p_i = 0$
    \EndIf
    \State $\widetilde{X}_i = X_i+E_i$
    \State $e_i = \|\widetilde{X}_i-\widetilde{X}_i^2\|_2$
    \If{$i\geq 2$ \textbf{and} $p_i \neq p_{i-1}$ \textbf{and} $\frac{\log(e_i/C_2^{\textrm{sp}})}{\log(e_{i-2})} < 1.8$}
    \State $n=i$
    \State \textbf{break}    
    \EndIf
    \EndFor
    \State \textbf{output:} $n$, $\widetilde{X}_n$
  \end{algorithmic}
\end{algorithm}
Knowledge about the homo
and lumo eigenvalues also makes it possible to work with a predefined
sequence of polynomials and employ a scale-and-fold acceleration
technique~\cite{Rub2011} as described in
Algorithms~\ref{alg:selection_polynomials_sp2acc}
and~\ref{alg:sp2acc}, respectively.
In Algorithm~\ref{alg:selection_polynomials_sp2acc},
$\lambda_\textrm{homo}^\textrm{out}$ and
$\lambda_\textrm{homo}^\textrm{in}$ are bounds of the homo eigenvalue
such that
\begin{equation}
\lambda_\textrm{homo}^\textrm{out} \leq
\lambda_\textrm{homo} \leq
\lambda_\textrm{homo}^\textrm{in},
\end{equation}
$\lambda_\textrm{lumo}^\textrm{in}$ and
$\lambda_\textrm{lumo}^\textrm{out}$ are bounds of the lumo eigenvalue
such that
\begin{equation}
\lambda_\textrm{lumo}^\textrm{in} \leq
\lambda_\textrm{lumo} \leq \lambda_\textrm{lumo}^\textrm{out},
\end{equation}
and $\varepsilon_M$ denotes the machine epsilon.

Using the scale-and-fold technique the eigenspectrum is
stretched out outside the $[0, \ 1]$ interval and folded
back by applying the SP2 polynomials. 
The unoccupied part of the
eigenspectrum is partially stretched out below 0 and folded back into
$[0, \ 1]$ using the polynomial $((1-\alpha)I + \alpha x)^2$, where
$\alpha \geq 1$ determines the amount of stretching or
acceleration. Similarly, the occupied part of the eigenspectrum is
stretched out above 1 and folded back using $2\alpha x - (\alpha
x)^2$.

\begin{algorithm}\caption{Determination of polynomials for SP2-ACC homo/lumo-based expansion\label{alg:selection_polynomials_sp2acc}}
  \begin{algorithmic}[1]  
    \State \textbf{input:} $\lambda_\textrm{min}$, $\lambda_\textrm{max}$, 
    $\lambda_\textrm{homo}^\textrm{out}$, $\lambda_\textrm{homo}^\textrm{in}$,
    $\lambda_\textrm{lumo}^\textrm{in}$, $\lambda_\textrm{lumo}^\textrm{out}$
    \State $\beta_\upperBound = 1-
    \frac{\lambda_\textrm{max} - \lambda_{\textrm{homo}}^\textrm{in}}
         {\lambda_\textrm{max} - \lambda_\textrm{min}}$
    \State $\beta_\lowerBound = 1-
    \frac{\lambda_\textrm{max} - \lambda_{\textrm{homo}}^\textrm{out}}
         {\lambda_\textrm{max} - \lambda_\textrm{min}}$ 
         \label{alg:selection_polynomials_sp2acc_beta_low}
    \State $\gamma_\upperBound = 
    \frac{\lambda_\textrm{max} - \lambda_{\textrm{lumo}}^\textrm{in}}
         {\lambda_\textrm{max} - \lambda_\textrm{min}}$
    \State $\gamma_\lowerBound = 
    \frac{\lambda_\textrm{max} - \lambda_{\textrm{lumo}}^\textrm{out}}
         {\lambda_\textrm{max} - \lambda_\textrm{min}}$ 
         \label{alg:selection_polynomials_sp2acc_gamma_low}
    \State $\delta = 0.01$
    \State $i = 0$
    \While {($\beta_\upperBound-\beta_\upperBound^2 > \varepsilon_M $ or $\gamma_\upperBound-\gamma_\upperBound^2 > \varepsilon_M$) and $p_i = p_{i-1}$ }
    \State $i = i + 1$
    \If{$\beta_\lowerBound < \delta$ \textbf{and} $\gamma_\lowerBound < \delta$ }  \label{alg:selection_polynomials_sp2acc_delta_start}
    \State $\beta_\lowerBound = 0$, $\gamma_\lowerBound = 0$
    \State $n_{\textrm{min}} = i+1$
    \State $\delta = 0$
    \EndIf     \label{alg:selection_polynomials_sp2acc_delta_end}
    \If{$\gamma_\upperBound \geq \beta_\upperBound$}
    \State $p_{i} = 1$
    \State $\alpha_{i} = 2/(2-\gamma_\lowerBound)$ \label{alg:selection_polynomials_sp2acc_alpha1}

    \State $\gamma_\lowerBound = ((1-\alpha_{i})+\alpha_{i}\gamma_\lowerBound)^2,
    \quad
    \gamma_\upperBound = ((1-\alpha_{i})+\alpha_{i}\gamma_\upperBound)^2$
    \State $\beta_\lowerBound = 2\alpha_i\beta_\lowerBound-(\alpha_i\beta_\lowerBound)^2,
    \quad
    \beta_\upperBound = 2\alpha_i\beta_\upperBound-(\alpha_i\beta_\upperBound)^2$

    \Else
    \State $p_{i} = 0$ 
    \State $\alpha_{i} = 2/(2-\beta_\lowerBound)$ \label{alg:selection_polynomials_sp2acc_alpha2}

    \State $\gamma_\lowerBound = 2\alpha_i\gamma_\lowerBound-(\alpha_i\gamma_\lowerBound)^2,
    \quad
    \gamma_\upperBound = 2\alpha_i\gamma_\upperBound-(\alpha_i\gamma_\upperBound)^2$
    \State $\beta_\lowerBound = ((1-\alpha_{i})+\alpha_{i}\beta_\lowerBound)^2,
    \quad
    \beta_\upperBound = ((1-\alpha_{i})+\alpha_{i}\beta_\upperBound)^2$


    \EndIf
    \EndWhile
    \State $n_{\textrm{max}} = i$ 
    \State \textbf{output:} $n_{\textrm{min}}$, $n_{\textrm{max}}$, $p_i$, $\alpha_i$, $i=1,2,\dots,n_{\textrm{max}}$
  \end{algorithmic}
\end{algorithm}

\begin{algorithm}\caption{SP2-ACC algorithm\label{alg:sp2acc}}
  \begin{algorithmic}[1]
      \State \textbf{input:} $F$, $\lambda_\textrm{min}$, $\lambda_\textrm{max}$, $n_{\textrm{min}}$, $n_{\textrm{max}}$, $p_i$, $\alpha_i$, $i=1,2,\dots,n_{\textrm{max}}$
    \State $X_0 = 
    \frac{\lambda_\textrm{max}I - F}
         {\lambda_\textrm{max} - \lambda_\textrm{min}}$
    \State $\widetilde{X}_0 = X_0+E_0$
    \State $e_0 = \|\widetilde{X}_0-\widetilde{X}_0^2\|_2$
    \State $C_2^{\textrm{sp}} =\frac{1}{32} \left(71+17 \sqrt{17}\right)$
    \For{$i = 1,2,\dots, n_{\textrm{max}}$}
    \If{$p_i = 1$}
    \State $X_i = ((1-\alpha_i)I+\alpha_i\widetilde{X}_{i-1})^2$
    \Else
    \State $X_i = 2\alpha_i\widetilde{X}_{i-1}-(\alpha_i\widetilde{X}_{i-1})^2$
    \EndIf
    \State $\widetilde{X}_i = X_i+E_i$
    \State $e_i = \|\widetilde{X}_i-\widetilde{X}_i^2\|_2$
    \If{$i \geq  n_{\textrm{min}}$ \textbf{and} $p_i \neq p_{i-1}$ \textbf{and} $\frac{\log(e_i/C_2^{\textrm{sp}})}{\log(e_{i-2})} < 1.8$}
    \State $n=i$
    \State \textbf{break}    
    \EndIf
    \EndFor
    \State \textbf{output:} $n$, $\widetilde{X}_n$
  \end{algorithmic}
\end{algorithm}

In the following we will make use of Theorem~\ref{Cq_theorem} to
derive the stopping criteria employed in Algorithms~\ref{alg:sp2}
and~\ref{alg:sp2acc}.  We first note that neither of the limits
\begin{align} 
    \lim_{x \rightarrow 1^-} \frac{p_{\textrm{sp}}^{(1)}(x)-p_{\textrm{sp}}^{(1)}(x)^2}{(x-x^2)^2} \quad \textrm{and} \quad
    \lim_{x \rightarrow 0^+} \frac{p_{\textrm{sp}}^{(2)}(x)-p_{\textrm{sp}}^{(2)}(x)^2}{(x-x^2)^2}
\end{align}
exist.

The limits required by Theorem~\ref{Cq_theorem} do exist for
compositions of alternating polynomials $p_{\textrm{sp}}^{(12)}(x) =
p_{\textrm{sp}}^{(1)}(p_{\textrm{sp}}^{(2)}(x))$ and
$p_{\textrm{sp}}^{(21)}(x) =
p_{\textrm{sp}}^{(2)}(p_{\textrm{sp}}^{(1)}(x))$:
\begin{align} 
    \lim_{x \rightarrow 0^+} \frac{p_{\textrm{sp}}^{(21)}(x)-p_{\textrm{sp}}^{(21)}(x)^2}{(x-x^2)^2} = \lim_{x \rightarrow 1^-} \frac{p_{\textrm{sp}}^{(12)}(x)-p_{\textrm{sp}}^{(12)}(x)^2}{(x-x^2)^2} = 2, \\
    \lim_{x \rightarrow 1^-} \frac{p_{\textrm{sp}}^{(21)}(x)-p_{\textrm{sp}}^{(21)}(x)^2}{(x-x^2)^2} = \lim_{x \rightarrow 0^+} \frac{p_{\textrm{sp}}^{(12)}(x)-p_{\textrm{sp}}^{(12)}(x)^2}{(x-x^2)^2} = 4.
\end{align}
However, the limits
\begin{align} 
    \lim_{x \rightarrow 0^+} \frac{p_{\textrm{sp}}^{(22)}(x)-p_{\textrm{sp}}^{(22)}(x)^2}{(x-x^2)^2} \quad \textrm{and} \quad
    \lim_{x \rightarrow 1^-} \frac{p_{\textrm{sp}}^{(11)}(x)-p_{\textrm{sp}}^{(11)}(x)^2}{(x-x^2)^2}
\end{align}
do not exist.
We therefore want to find the smallest value $C_2^{\textrm{sp}}$ such that 
 \begin{equation}
  \frac{\log(e_i/C_2^{\textrm{sp}})}{\log(e_{i-2})} \geq 2,
\end{equation} 
provided that $p_i \neq p_{i-1}$.
Consequently, we  invoke Theorem~\ref{Cq_theorem} with $f = p_{\textrm{sp}}^{(12)}$ and $f = p_{\textrm{sp}}^{(21)}$, $a = 0$, and $q = 2$. Noting that 
\begin{align} 
     \max_{x \in [0, 1]} \frac{p_{\textrm{sp}}^{(12)}(x)-p_{\textrm{sp}}^{(12)}(x)^2}{(x-x^2)^2} 
     &=\max_{x \in [0, 1]} \frac{p_{\textrm{sp}}^{(21)}(x)-p_{\textrm{sp}}^{(21)}(x)^2}{(x-x^2)^2} \\
     &=\frac{1}{32} \left(71+17 \sqrt{17}\right)
\end{align}
we get 
\begin{align} 
    C_2^{\textrm{sp}} =\frac{1}{32} \left(71+17 \sqrt{17}\right) \approx 4.40915
\end{align}
which leads to the stopping criterion of Algorithm~\ref{alg:sp2}.

Theorem~\ref{Cq_theorem} can also be used to derive a stopping
criterion for the accelerated algorithm. However, in order for the
limits~\eqref{eq:limAB} to exist the parameter $a$ should be chosen
larger than 0 and vary throughout the iterations. Since the
acceleration is effective only in the conditioning phase, we have
found it easier to turn off the acceleration when entering the
purification phase and then use the stopping criterion for the regular
expansion.  We turn off the acceleration as soon as the relevant homo
and lumo eigenvalue bounds are close enough to 1 and 0 respectively,
and start to check the stopping criterion in the next iteration
$n_{\textrm{min}}$, see lines
\ref{alg:selection_polynomials_sp2acc_delta_start}--\ref{alg:selection_polynomials_sp2acc_delta_end}
of Algorithm~\ref{alg:selection_polynomials_sp2acc}.

\begin{figure}[!ht]
        \centering
        \includegraphics[width=0.6\textwidth]{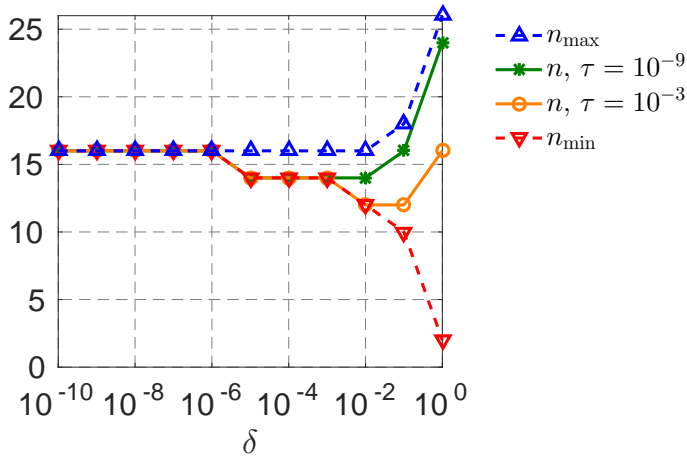}
        \caption{The estimated number of iterations
          $n_{\textrm{max}}$, iteration $n_{\textrm{min}}$ where the
          acceleration has been turned off, and total number of
          iterations $n$ in the SP2-ACC expansion for various values
          of $\delta$. The recursive expansion is applied to the
          matrix $X_0 = \text{diag}(0.48, 0.52)$ perturbed in each
          iteration by a diagonal matrix with random elements from a
          normal distribution and with spectral norm $\tau$.
        }
        \label{fig:optimal_delta}
\end{figure}
There is no need for accurate detection of the transition between the
conditioning and purification phases---the parameter $\delta$ in the
algorithm is to some extent arbitrary.
A larger $\delta$-value results in less acceleration. A smaller value
means that we will start to check the stopping criterion later
possibly resulting in superfluous iterations, particularly in low
accuracy calculations, see Figure~\ref{fig:optimal_delta}. By
numerical experiments we have found $\delta = 0.01$ to be an
appropriate value.  For values smaller than $0.01$ the effect of the
acceleration is less than 1 percent compared to the regular iteration,
see Figure~\ref{fig:acc_efficiency}.
\begin{figure}[!ht]
\centering
\includegraphics[width=0.5\textwidth]{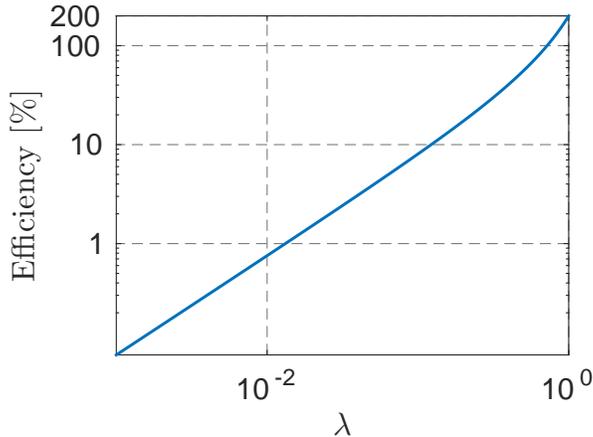}
\caption{ Efficiency of the accelerated SP2-ACC scheme relative to the
  regular SP2 scheme. The figure shows how much more an unoccupied
  eigenvalue $\lambda$ is reduced by the polynomial $(1-\alpha +
  \alpha \lambda)^2$ with $\alpha=2/(2-\lambda)$ compared to the
  polynomial $\lambda^2$, in the relative sense,
  i.e. $|\lambda^2-(1-\alpha + \alpha
  \lambda)^2|/|\lambda-\lambda^2|$. The corresponding (identical)
  figure for an occupied eigenvalue can be constructed in the same
  way.  Clearly the acceleration has almost no effect in the
  purification phase when all eigenvalues are close to their desired
  values of~0 or~1. The plot is in log-log scale.
  \label{fig:acc_efficiency}}
\end{figure}

\subsection{Estimation of the order of convergence}
As for the McWeeny expansion one may want to use the Frobenius norm instead of the spectral norm to measure the idempotency error.
However, for the SP2 expansion, relation \eqref{eq:K_ratio} translates to
\begin{equation} \label{eq:K_ratio_SP2}
  K_{i-2}^2/K_i \geq 1, \quad \textrm{if } p_i \neq p_{i-1} 
\end{equation}
given second order convergence and the application of
Theorem~\ref{Cq_theorem} for compositions of alternating polynomials
from two iterations.  We have not been able to prove that
\eqref{eq:K_ratio_SP2} always holds. However, we have also not been
able to find any counterexample where \eqref{eq:K_ratio_SP2} does not
hold in exact arithmetics.

We have encountered cases when \eqref{eq:K_ratio_SP2} does not hold
due to numerical errors. However, the use of the Frobenius norm has
not resulted in too early stops in those cases.  An example is given
in Figure~\ref{fig:dense_1000_500} where we applied the recursive
expansion to a random symmetric dense matrix. The occupied and
unoccupied eigenvalues were distributed equidistantly in $[0,
  \ 0.495]$ and $[0.505, \ 1]$, respectively. The eigenvectors of the
matrix were taken from a QR factorization of a matrix with random
elements from a normal distribution.
The use of the Frobenius norm in 
the stopping criterion results in a stop in iteration 29 while the use
of the spectral norm results in a stop in iteration 31. However, being
clear from panels (c) and (d), the stagnation phase started already in
iteration 29.  Although \eqref{eq:K_ratio_SP2} does not hold, the use
of the Frobenius norm does not result in a too early stop in this
case.

\begin{figure}[ht!]
        \centering
        \captionsetup[subfigure]{justification=centering}
        \begin{subfigure}[b]{0.45\textwidth}
                \includegraphics[width=\textwidth]{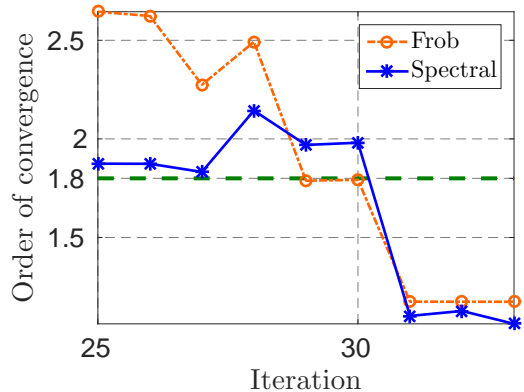}
                \caption{Observed and estimated orders of convergence}
                \label{fig:dense_1000_500_orders}
        \end{subfigure}
        \begin{subfigure}[b]{0.45\textwidth}
                \includegraphics[width=\textwidth]{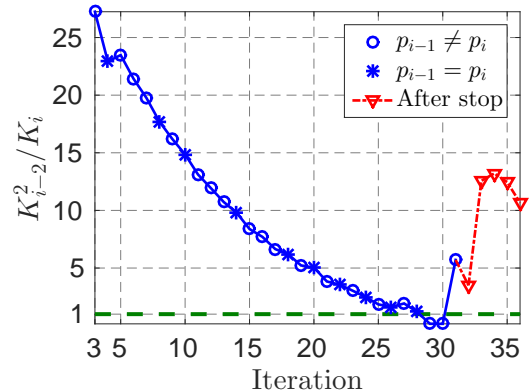}
                \caption{${K_{i-2}^{2}}/{K_{i}}$\\ \quad} 
                \label{fig:dense_1000_500_Kratio}
        \end{subfigure}

        \begin{subfigure}[b]{0.45\textwidth}
		\includegraphics[width=\textwidth]{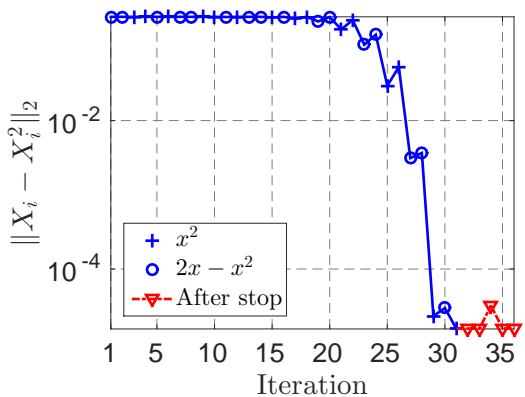}
                \caption{Idempotency error using \\spectral norm}
                \label{fig:dense_1000_500_sp_norm}
        \end{subfigure}
        \begin{subfigure}[b]{0.45\textwidth}
		\includegraphics[width=\textwidth]{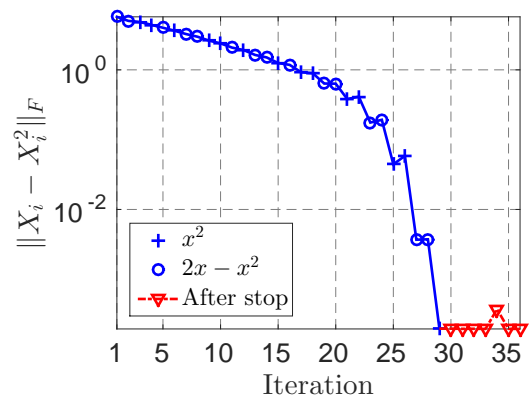}
                \caption{Idempotency error using \\Frobenius norm}
                \label{fig:dense_1000_500_frob_norm}
        \end{subfigure}
        \caption{Example illustrating a special case when use of the
          Frobenius norm in the stopping criterion results in an
          earlier stop than use of the spectral norm. The regular SP2
          expansion is applied to a random symmetric dense $1000 \times 1000$
          matrix with all eigenvalues in $[0,\ 1]$, homo-lumo gap 0.01
          located symmetrically around 0.5, and occupation number
          500. In each iteration $i$ the matrix $X_i$ is perturbed by
          a random symmetric matrix with elements from a normal
          distribution and with spectral norm $\e{-5}$.
}
        \label{fig:dense_1000_500}
\end{figure}

\subsection{Number of subsequent iterations with the same polynomial}
\label{sec:num_iter}

The stopping criteria in Algorithms~\ref{alg:sp2} and~\ref{alg:sp2acc}
include a condition of alternating polynomials, i.e. that $p_{i} \neq
p_{i-1}$ in iteration $i$.  When the polynomials in each iteration are
chosen based on the trace of the matrix as in Algorithm~\ref{alg:sp2},
it is possible to construct examples where the same polynomial appears
in an arbitrary number of subsequent iterations.  In particular, you
will get a large number of consecutive iterations with the same
polynomial if there are many eigenvalues clustered very close to
either the homo or the lumo eigenvalue. However, besides artificially
constructed matrices we have not come across Fock or Kohn--Sham
matrices that give more than a few subsequent iterations with the same
polynomial.  When the polynomials are chosen based on the location of
the homo and lumo eigenvalues as in
Algorithm~\ref{alg:selection_polynomials_sp2acc}, the number of
subsequent iterations is strictly bounded.  When
Algorithm~\ref{alg:selection_polynomials_sp2acc} is used without
acceleration there can after an initial startup phase be at most two
subsequent iterations with the same polynomial, as shown by the
following theorem.

\begin{theorem} \label{tm:num_iter_sp2}
  In Algorithm~\ref{alg:selection_polynomials_sp2acc}, let
  $\beta_\upperBound \neq 0$, $\gamma_\upperBound\neq 0$,
  $\beta_\lowerBound = 0$, and $\gamma_\lowerBound = 0$.  Assume that
  $p_{i} \neq p_{i-1}$. Then, if $p_{i+1} = p_{i}$ it follows that $p_{i+2} \neq
  p_{i+1}$.
\end{theorem}

\begin{proof}

We use here the notation
\begin{align}
 \beta_{i} \coloneqq \beta_\upperBound \quad \textrm{and} \quad 
 \gamma_{i} \coloneqq \gamma_\upperBound
\end{align}
in every iteration $i$ of the recursive expansion.  Without loss of
generality we consider the case $i=2$. Assume that $p_{1} = 1$.  Then,
$\beta_{0} \leq \gamma_{0}$ and the largest possible number of
subsequent iterations with $p_k = 0, \ k \geq 2$ is obtained with
$\beta_{0} = \gamma_{0}$. Then, following
Algorithm~\ref{alg:selection_polynomials_sp2acc} we have that
\begin{align}
     \gamma_{1}  & = \gamma_{0}^2, \\
    \beta_{1} & = 2\beta_{0} - \beta_{0}^2  =2\gamma_{0} - \gamma_{0}^2.
  \end{align}
  Since $\beta_{1} > \gamma_{1}$, $p_2 = 0$ and
  \begin{align}
    \gamma_{2}  & = 2 \gamma_{1} - \gamma_{1}^2 = 2 \gamma_{0}^2 - \gamma_{0}^4, \\
    \beta_{2} & = \beta_{1}^2 = (2\gamma_{0} - \gamma_{0}^2)^2 = 2 \gamma_{0}^2 - \gamma_{0}^4 + 2\gamma_{0}^2(\gamma_{0}-1)^2.
  \end{align}
  Then, since $2\gamma_{0}^2(\gamma_{0}-1)^2 > 0$, we have that $\beta_{2} >
  \gamma_{2}$ and $p_3 = 0$. Therefore
  \begin{align}
    \gamma_{3} & = 2 \gamma_{2} - \gamma_{2}^2 = 2(2 \gamma_{0}^2 -
    \gamma_{0}^4) - (2 \gamma_{0}^2 - \gamma_{0}^4)^2 
    = 4\gamma_{0}^2 - 6\gamma_{0}^4 + 4\gamma_{0}^6 - \gamma_{0}^8,
\\
    \beta_{3} & = \beta_{2}^2 = (2\gamma_{0} - \gamma_{0}^2)^4 \nonumber \\
    & =
    \underbrace{4\gamma_{0}^2 - 6\gamma_{0}^4 + 4\gamma_{0}^6 -
      \gamma_{0}^8}_{\gamma_3}
    \underbrace{-4\gamma_{0}^2(2-11\gamma_{0}^2+16\gamma_{0}^3-10\gamma_{0}^4+4\gamma_{0}^5-\gamma_{0}^6)}_{<0}.
  \end{align}
  Thus, $\beta_{3} < \gamma_{3}$ and $p_4 = 1$. The case with $p_1 =
  0$ can be shown similarly.
  \qed
\end{proof}

Note that Theorem~\ref{tm:num_iter_sp2} does not exclude the
possibility of a large number of initial iterations with the same
polynomial. However, as soon as each of the two polynomials has been
used at least once, there will not be more than two subsequent
iterations with the same polynomial.  For
Algorithm~\ref{alg:selection_polynomials_sp2acc} with acceleration, it
is possible to show that there cannot be more than three subsequent
iterations with the same polynomial.

\section{Numerical experiments}
\label{sec:numerical_experiments}
This section provides numerical illustrations of the proposed stopping
criteria showing that they work well for dense and sparse matrices
regardless of what method is used to select matrix elements for
removal. All tests in this section are performed in Matlab R2015b. We
will here use the regular SP2 expansion without acceleration. Similar
results can be shown for expansions with other polynomials, such as
McWeeny or accelerated SP2.  It will be assumed that non-overlapping
intervals containing the homo and lumo eigenvalues, respectively, are
known before the start of the expansion so that
Algorithm~\ref{alg:selection_polynomials_sp2acc} can be used for
selection of polynomials and Algorithm~\ref{alg:sp2acc} for the
expansion. To get a regular (nonaccelerated) expansion the parameters
$\beta_\lowerBound$ and $\gamma_\lowerBound$ are both set to 0 on
lines \ref{alg:selection_polynomials_sp2acc_beta_low} and
\ref{alg:selection_polynomials_sp2acc_gamma_low} of
Algorithm~\ref{alg:selection_polynomials_sp2acc} giving scaling
parameters $\alpha_i, i = 1,2,\dots$ on lines
\ref{alg:selection_polynomials_sp2acc_alpha1} and
\ref{alg:selection_polynomials_sp2acc_alpha2} equal to 1 and
$n_{\textrm{min}}$ equal to 2.

In our first test we apply the SP2 expansion to a random symmetric
dense matrix, see Figure~\ref{fig:idemp_error_random_dense}. The
eigenvectors of the matrix were taken from a QR factorization of a
matrix with random elements from a normal distribution.  The occupied
and unoccupied eigenvalues were distributed equidistantly in $[0, \
  0.49]$ and $[0.51, \ 1]$, respectively. In each iteration, the matrix
$X_i$ was perturbed by a random symmetric matrix with spectral norm
$\tau$ and elements from a normal distribution.  The spectral norm was
used to compute the observed order of convergence used in the stopping
criterion.  The figure shows that the stopping criterion accurately
detects when numerical errors start to dominate the calculation and
prevent any further improvement of the eigenvalues.

\begin{figure}
        \centering
         \includegraphics[width=0.55\textwidth]{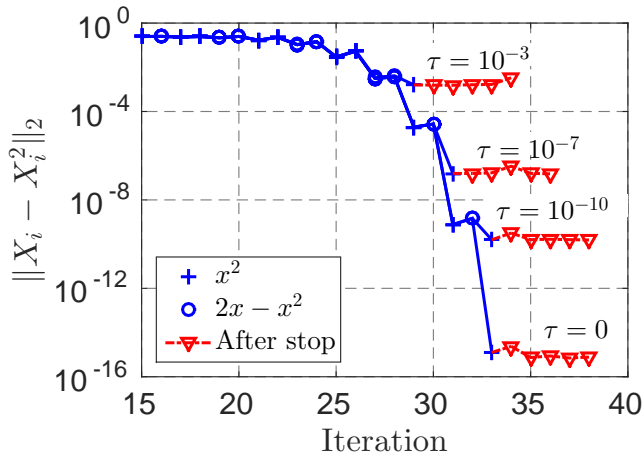}
        \caption{Idempotency error in each iteration of the SP2
          expansion for a random symmetric dense $200 \times 200$
          matrix with spectrum in $[0, \ 1]$ and homo-lumo gap 0.02 located
          around 0.5 and occupation number 100. In each iteration the
          matrix was perturbed by a random symmetric matrix with
          spectral norm $\tau$.  The idempotency error is shown for 4
          different values of $\tau$, before (blue circles and plus
          signs) and after (red triangles) the stop.
          \label{fig:idemp_error_random_dense}}
\end{figure}

We show in Figure~\ref{fig:num_test_idemp_error_trunc} that the
stopping criteria work well for two different approaches to select
matrix elements for removal in a linear scaling sparse matrix
setting. The Fock matrix comes from a converged spin-restricted
Hartree--Fock calculation for a linear alkane molecule
C$_{160}$H$_{322}$ using the standard Gaussian basis set STO-3G,
giving a total of 1122 basis functions and 641 occupied orbitals. In
this case the eigenvalue problem \eqref{eq:eig_problem} is on
generalized form
\begin{equation}
  Fx_i = \lambda_i Sx_i,
\end{equation}
where $S$ is the basis set overlap matrix. A congruence transformation
employing the inverse Cholesky factor of $S$ is used to get the
eigenvalue problem on standard form as required by the algorithms
described in this article. There are three different common approaches for
removal of matrix elements. Each matrix element in a Fock or
density matrix usually corresponds to the distance between two atomic
nuclei.  In cutoff radius based truncation, all elements corresponding
to distances larger than a predefined cutoff radius are removed.
However, in methods employing the congruence transformation this
approach cannot be straightforwardly applied.  In element magnitude
based truncation, all elements with absolute value smaller than a
predefined threshold value are removed.
The threshold value is typically chosen based on practical experience
without being directly linked to the actual error in the final result,
the density matrix.  However, a surprisingly simple relationship
between the truncation and the error in the final result was developed
by Rubensson \etal~\cite{m-accPuri}, allowing us to control the
forward error $\|D-\widetilde{X}_n\|_2$. The forward error is split in
two parts, the error in eigenvalues (idempotency error) and the error
in the occupied subspace. The subspace error is due to numerical
errors and the eigenvalue error is due to numerical errors and a
finite number of iterations.
Figure~\ref{fig:num_test_idemp_error_trunc_magnitude} shows that our
stopping criteria work well together with truncation based on matrix
element magnitude.
Figure~\ref{fig:num_test_idemp_error_trunc_subspace} shows that our
stopping criteria work well when matrix elements are removed with
control of the subspace error~\cite{m-accPuri}.

\begin{figure}[ht!]
        \centering
        \captionsetup[subfigure]{justification=centering}
	\begin{subfigure}[b]{0.45\textwidth}
	  \includegraphics[width=\textwidth]{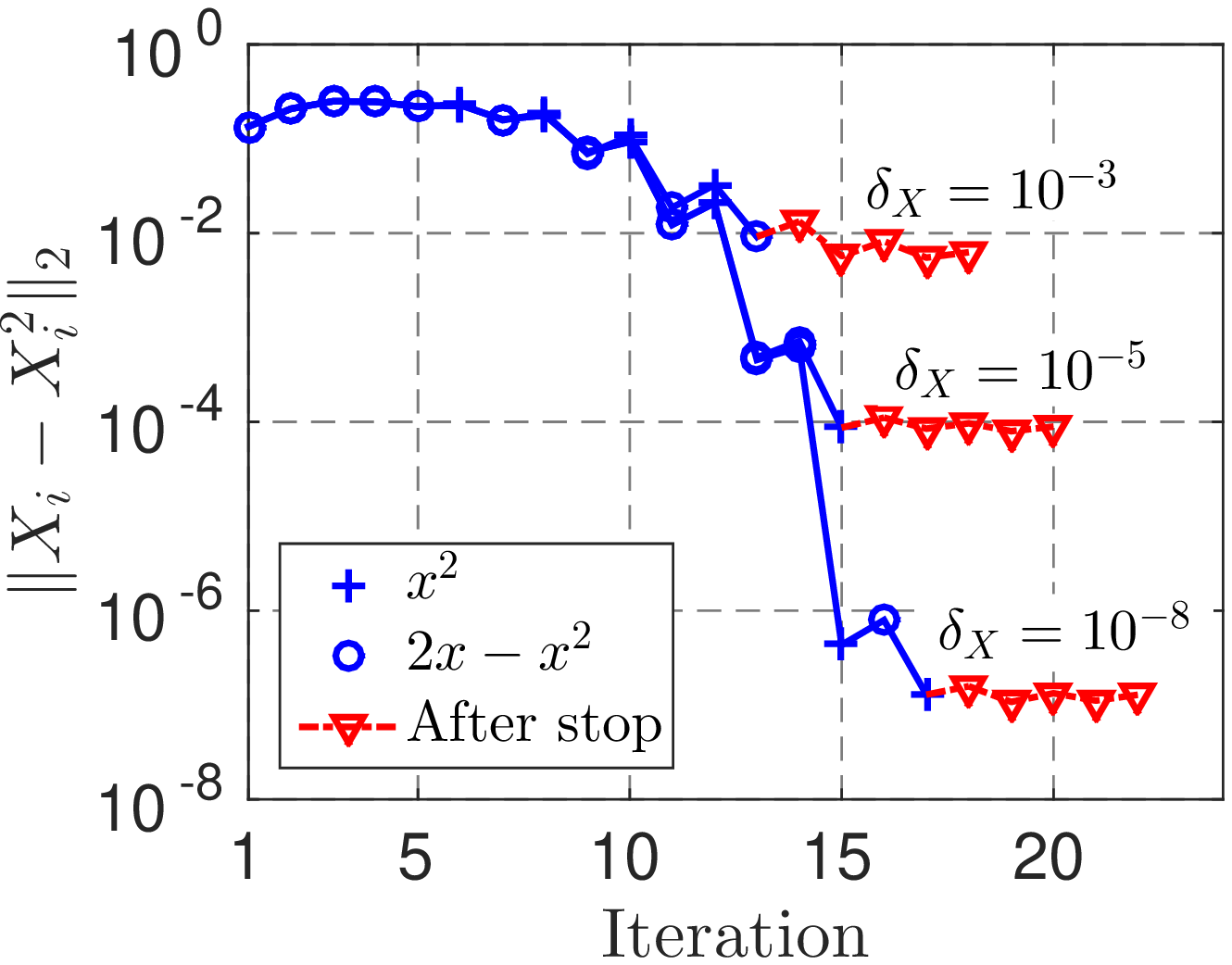}
	  \caption{Truncation based on the \\magnitude relation \label{fig:num_test_idemp_error_trunc_magnitude}}
	\end{subfigure}
	\begin{subfigure}[b]{0.45\textwidth}
	  \includegraphics[width=\textwidth]{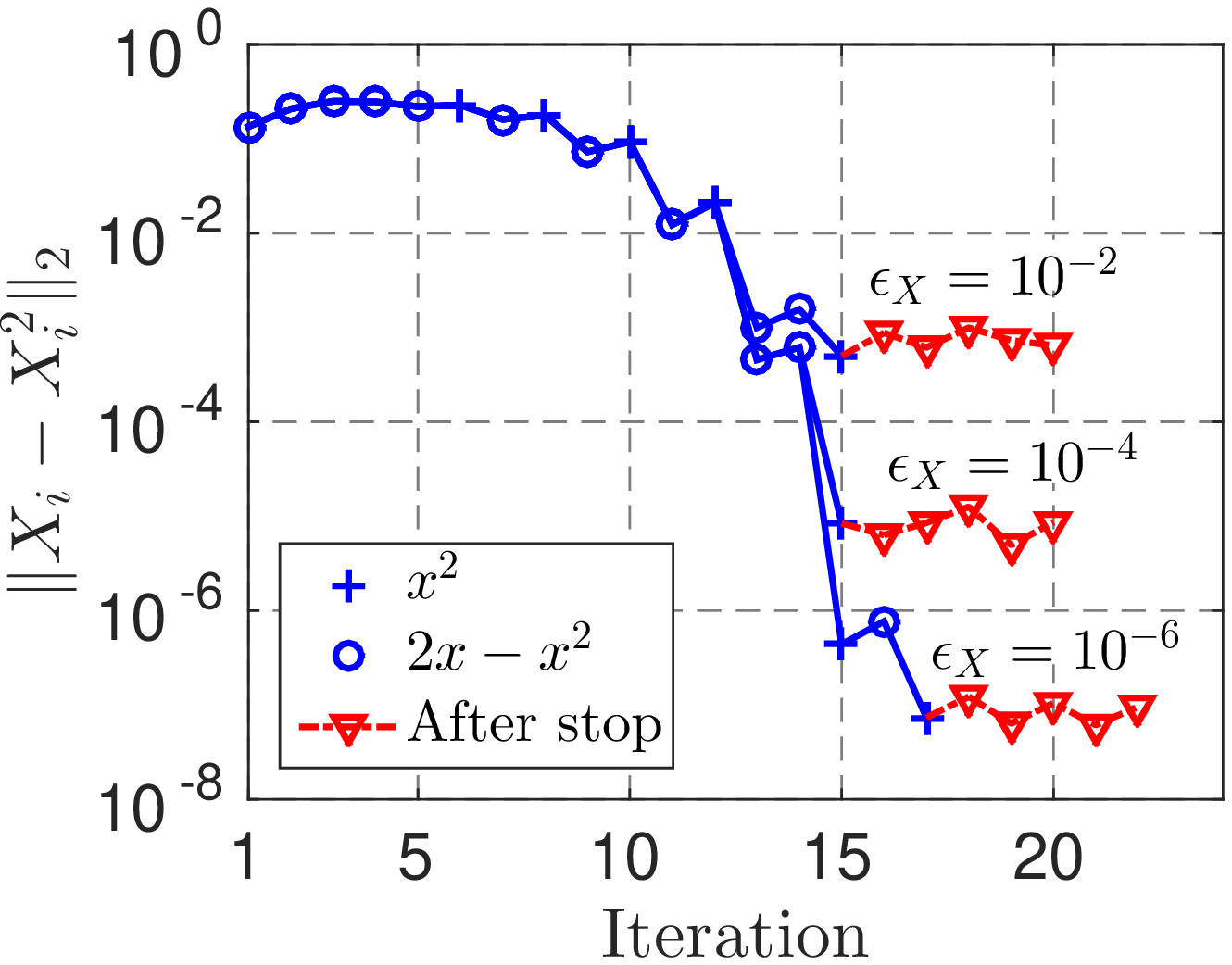}
	  \caption{Truncation based on the \\spectral norm of the error matrix \label{fig:num_test_idemp_error_trunc_subspace}}
	\end{subfigure}
        \caption{Demonstration of the stopping criterion with
          different ways of selecting matrix elements for removal in a
          sparse matrix setting.  The idempotency error in each
          iteration of the SP2 expansion is shown for a Fock 
          matrix coming from a Hartree--Fock calculation
          on a linear alkane using a standard Gaussian basis set.
          Panel~(a): Magnitude based truncation with threshold value
          $\delta_X$ for removing small matrix elements.  Panel~(b):
          Removal of elements with control of the error in the
          occupied subspace with the requirement that $\sin \theta <
          \epsilon_X$ where $\theta$ is the largest canonical angle
          between the exact and perturbed subspaces.
        }
        \label{fig:num_test_idemp_error_trunc}
\end{figure}

We will now investigate how the stopping criterion works when the
observed order is estimated using the Frobenius and mixed matrix
norms. We apply the SP2 expansion to a diagonal matrix of dimension
$\e{8} \times \e{8}$ with occupation number $\e{8}/2$, homo-lumo gap
0.1 located symmetrically around 0.5 and with otherwise equidistant
eigenvalues in $[0, \ 1]$.  The idempotency errors in each iteration
measured by the Frobenius, mixed, and spectral norms are shown in
Figure~\ref{fig:num_test_diag_matrices}.  The block size for the mixed
norm is 1000.  As anticipated in Section~\ref{sec:paramless}, $\|X_i -
X_i^2\|_F$ never goes below~1 for such a large system and
\eqref{eq:estim_order_using_frob_norm} cannot be used to estimate the
observed order. If the spectral norm is expensive to compute, the
mixed norm may be used in such cases.

\begin{figure}[ht!]
        \centering
	  \includegraphics[width=0.5\textwidth]{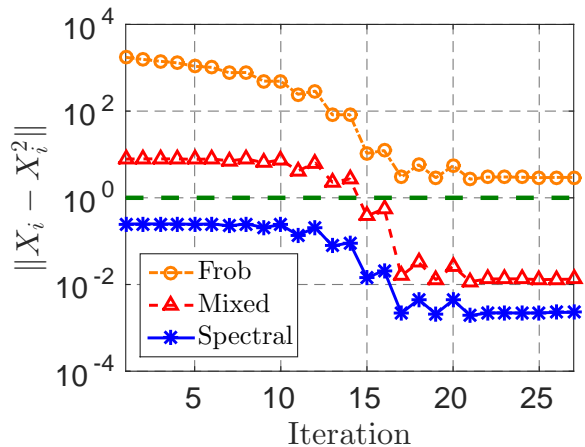} 
        \caption{Frobenius, mixed, and spectral norms of $X_i-X_i^2$
          in every iteration $i$ of the recursive SP2 expansion
          applied to a diagonal $\e{8} \times \e{8}$ matrix with
          occupation number $\e{8}/2$, homo-lumo gap 0.1 located
          symmetrically around 0.5, and otherwise equidistantly
          distributed eigenvalues in $[0,\ 1]$. In each iteration the
          matrix was perturbed by a diagonal matrix with random
          elements from a normal distribution and with spectral norm
          $\e{-3}$. The block size for the mixed norm is 1000.}
        \label{fig:num_test_diag_matrices}
\end{figure}


\section{Application to self-consistent field calculations}\label{sec:scf_calc}
In this section we use the developed stopping criteria in the regular
and accelerated SP2 expansions in self-consistent field (SCF)
calculations with the quantum chemistry program {\sc
  Ergo}~\cite{ergo_web,Ergo2011}. We have performed spin-restricted
Hartree--Fock calculations on a cluster of 4158 water molecules using
the standard Gaussian basis set 3-21G, giving a total of 54054 basis
functions. As initial guess we used the result of a calculation with a
smaller basis set, STO-3G. We used direct inversion in the iterative
subspace (DIIS) for convergence acceleration~\cite{pulay_1980,pulay82}
and stopped the iterations as soon as the largest absolute element of
$FDS-SDF$ was smaller than $\e{-3}$. The hierarchical matrix library
~\cite{hierarchical_matrix_lib} was used for sparse matrix
operations. A block size of 32 was used at the lowest level in the
sparse hierarchical representation. The mixed norm with block size 32
was used both in the stopping criterion and for removal of small
matrix elements~\cite{mixedNormTrunc} with a tolerance of $10^{-3}$
for the error in the occupied subspace measured by the largest
canonical angle between the exact and approximate
subspaces~\cite{m-accPuri}.

The tests were running on the Tintin cluster at the UPPMAX computer
center in Uppsala University using the gcc 5.3.0 compiler and the
OpenBLAS~\cite{openblas} library was used for matrix operations at the
lowest level in the sparse hierarchical representation. Each node on
Tintin is a dual AMD Bulldozer compute server with two 8-core Opteron
6220 processors running at 3.0 GHz. The calculations presented here
are performed on a node with 128 GB of memory.

In each SCF cycle we compute upper and lower bounds of the homo and
lumo eigenvalues and propagate them to the next SCF cycle.  It was
shown by Rubensson and Zahedi~\cite{interior_eigenvalues_2008} that
eigenvalues around the homo-lumo gap can be computed efficiently by
making use of the ability of the recursive expansion to give large
separation between interior eigenvalues. In that work, the Lanczos
method was used to extract the desired information.  Here we use a
recent approach to compute accurate homo and lumo bounds that only
requires the evaluation of Frobenius norms and traces during the
course of the recursive expansion~\cite{interior_eigenvalues_2014}.
Intervals containing the homo and lumo eigenvalues are propagated
between SCF cycles using Weyl's theorem for eigenvalue
movement~\cite{interior_eigenvalues_2014,m-accPuri}.  The inner bounds
for the homo and lumo eigenvalues are used both for the determination
of polynomials in Algorithm~\ref{alg:selection_polynomials_sp2acc} and
for the error control. The outer bounds are used for the acceleration
in Algorithm~\ref{alg:selection_polynomials_sp2acc}. When inner bounds
for homo and lumo are not known or are not accurate we fall back to
the trace-correcting SP2 expansion described in
Algorithm~\ref{alg:sp2}. However, even if the outer bounds are loose
Algorithm~\ref{alg:selection_polynomials_sp2acc} can be used with a
modification that the polynomials are determined on the fly using the
condition in Line~\ref{line:cond_sp2} in Algorithm~\ref{alg:sp2}.

The regular and accelerated SP2 expansions are compared in
Figure~\ref{fig:water_SCF_acc_noacc}. 
\begin{figure}[ht!]
	\captionsetup[subfigure]{justification=centering}
        \centering
        \begin{subfigure}[t]{.45\textwidth}
        \centering
	  \includegraphics[width=\textwidth]{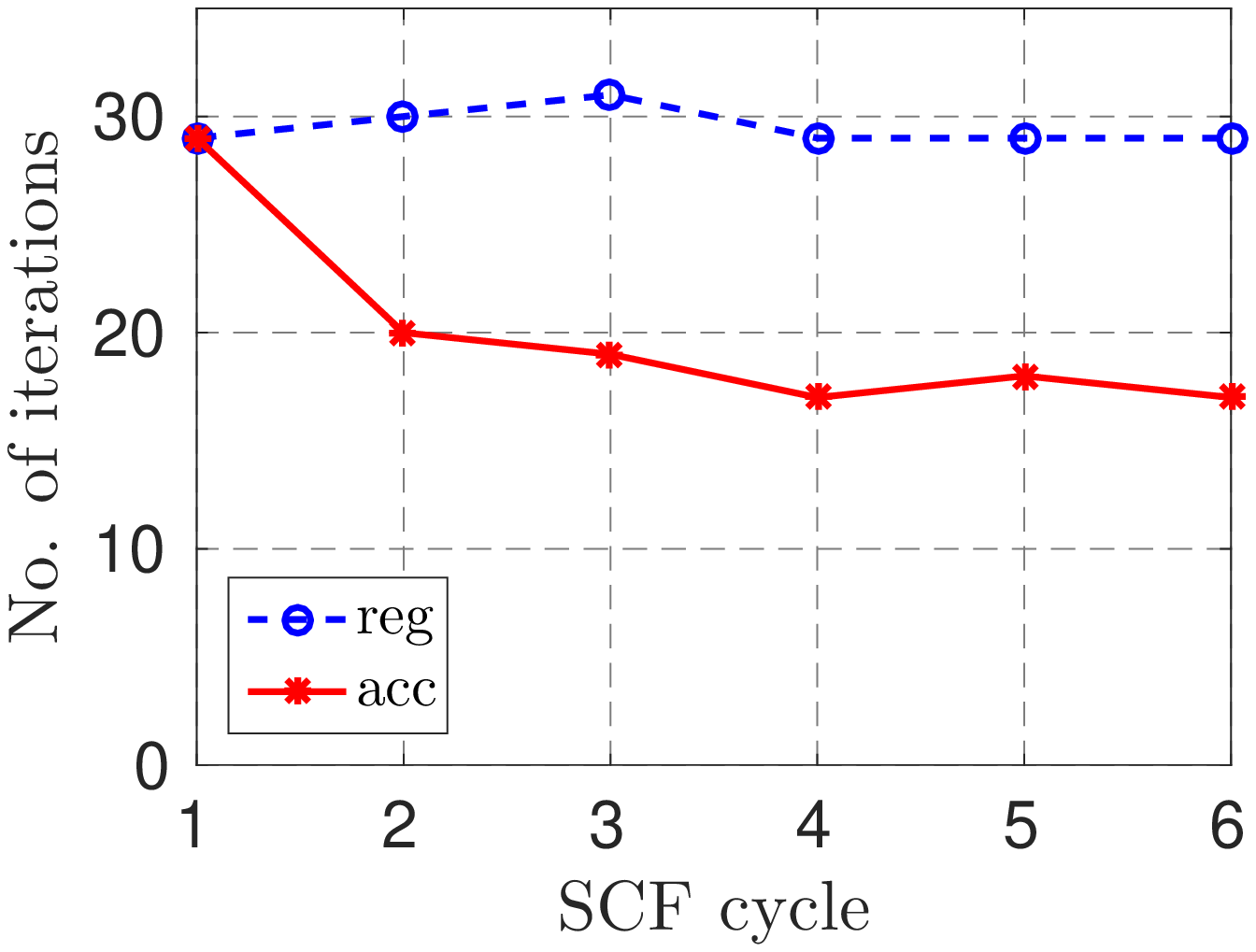}
	  \caption{Number of iterations in the recursive expansion}
	\end{subfigure}
	\begin{subfigure}[t]{.45\textwidth}
	\centering
	  \includegraphics[width=\textwidth]{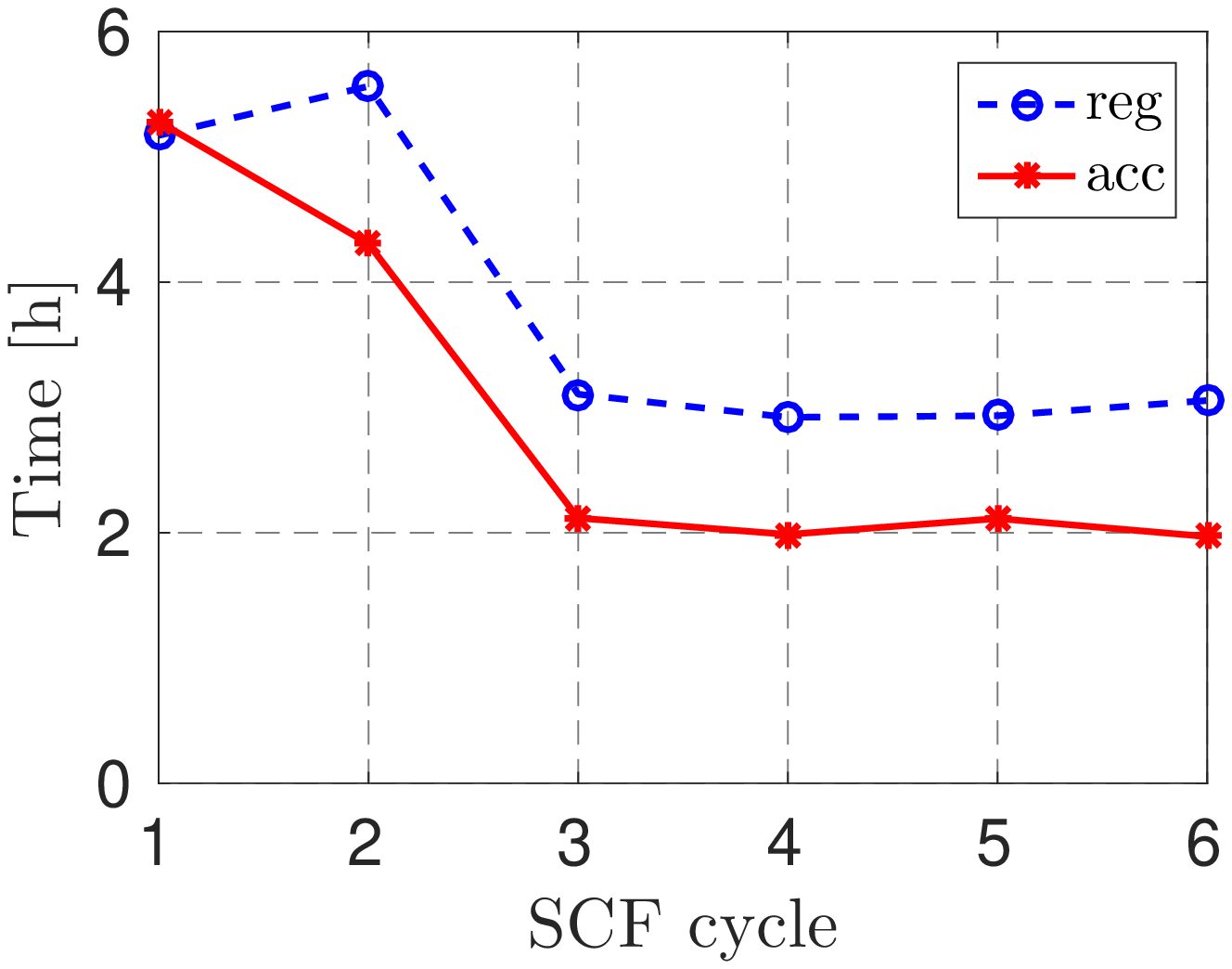}
	  \caption{Recursive expansion wall time}
	\end{subfigure}
      \hfill
	\begin{subfigure}[t]{0.45\textwidth}
	\centering
	  \includegraphics[width=\textwidth]{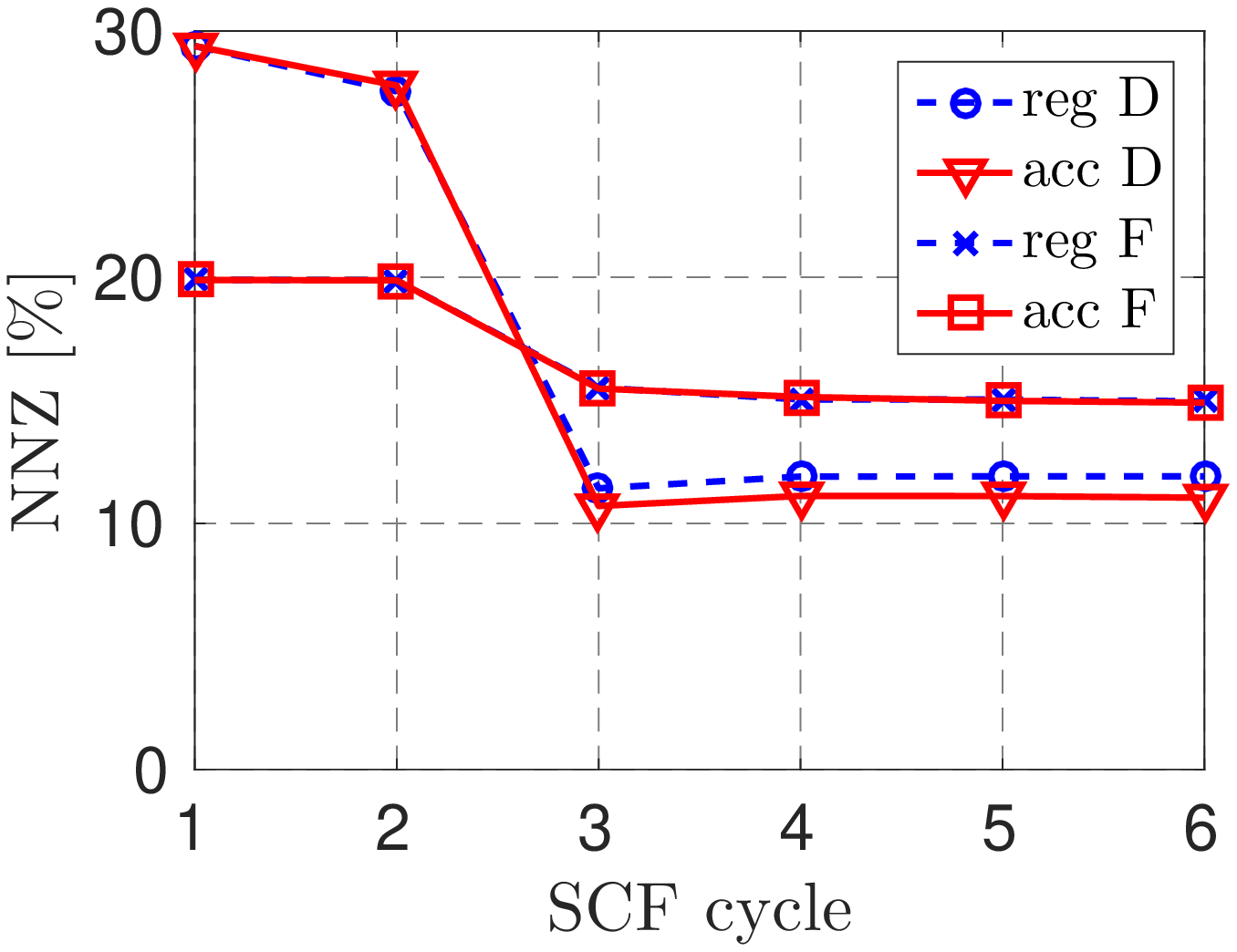}
	  \caption{Number of non-zeros}
	\end{subfigure}
        \caption{Comparison of the regular and accelerated SP2
          expansions in a self-consistent field Hartree--Fock
          calculation for a water cluster with 4158 water
          molecules. Panels~(a) and~(b) show the number of iterations
          and wall time of the recursive expansion in each
          self-consistent field cycle. Panel~(c) gives the
          percentage of non-zero elements in the Fock and density
          matrices, $F$ and $D$ respectively, in orthogonal basis. See
          the text for more details.}
        \label{fig:water_SCF_acc_noacc}
\end{figure}
Panels~(a) and~(b) show the
number of iterations and wall time, respectively, for the regular and
accelerated expansions in each SCF cycle. Panel~(c) shows the
percentage of non-zero elements in the Fock and density matrices in
orthogonal basis in each SCF cycle. Matrices were transformed from
non-orthogonal basis using the inverse Cholesky factor as in the
previous section. The number of non-zeros and the idempotency error
for the regular and accelerated SP2 expansions in the last SCF cycle
are shown in Figure~\ref{fig:water_SCF_acc_noacc_last_SCF_cycle}.  In
this case, the accelerated expansion takes a shorter but less sparse
route to the final density matrix. Thus, the peak memory usage is
larger and in some iterations significantly more work is required.
However, the acceleration gives a substantial reduction of the overall
computational time needed for the recursive expansion.

\begin{figure}[ht!]
        \centering
        \captionsetup[subfigure]{justification=centering}
        \begin{subfigure}[t]{.45\textwidth}
        \centering
	  \includegraphics[width=\textwidth]{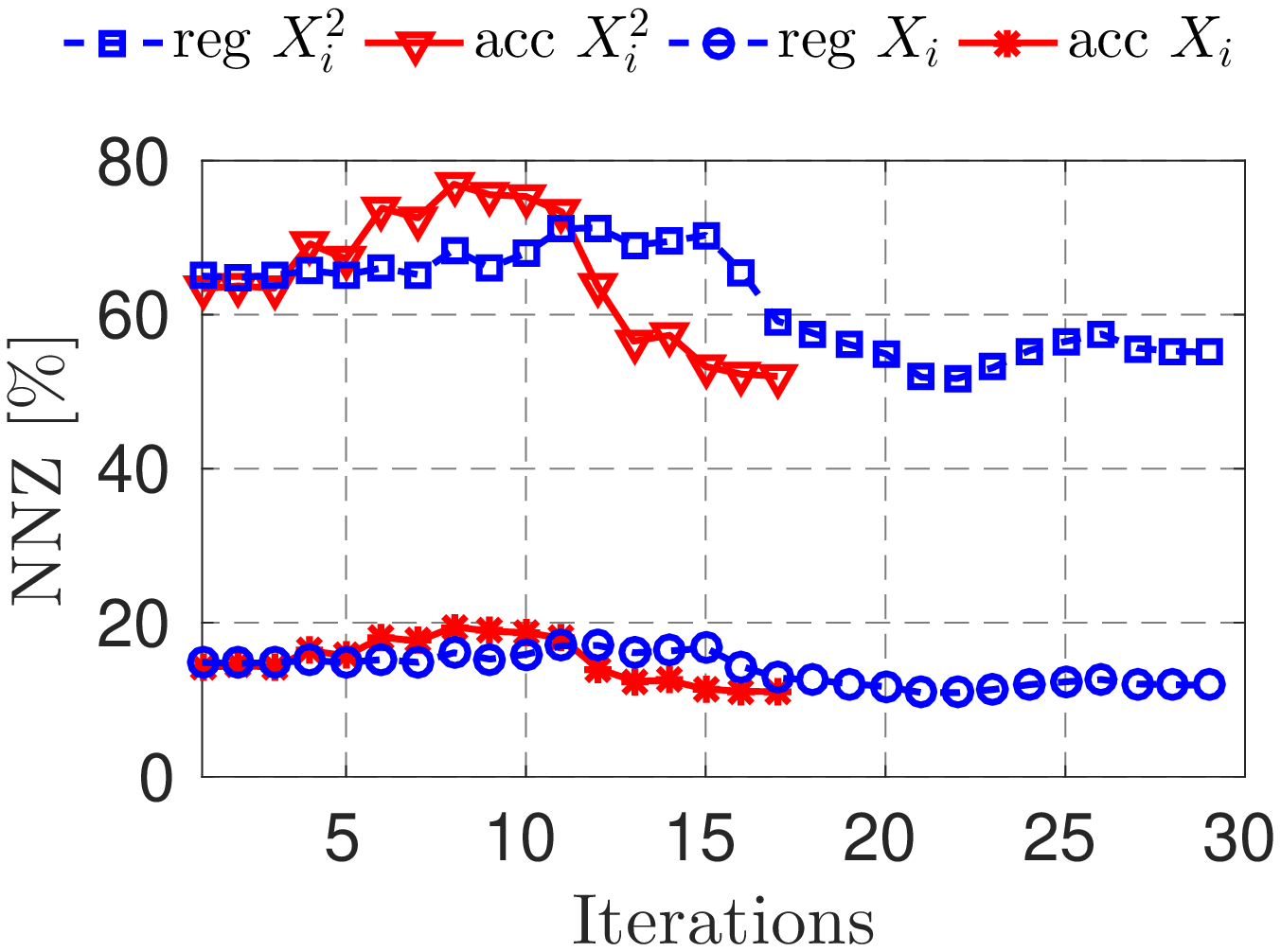}
	  \caption{Number of non-zeros}
	\end{subfigure}
	\begin{subfigure}[t]{.45\textwidth}
	\centering
	  \includegraphics[width=\textwidth]{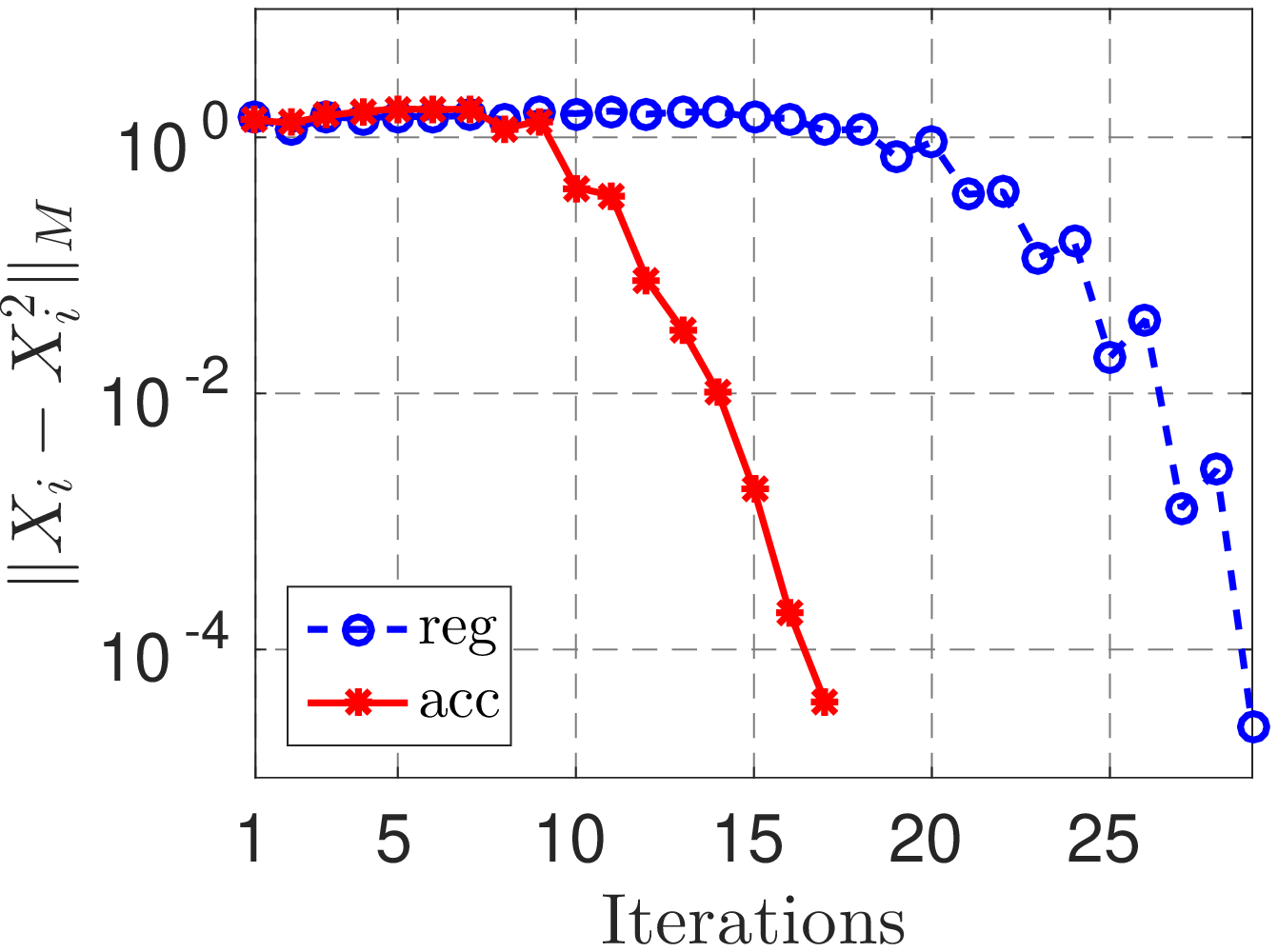}
	  \caption{Idempotency error measured using the mixed norm}
	\end{subfigure}
        \caption{Comparison of the regular and accelerated SP2
          expansions in the last cycle in a
          self-consistent field Hartree--Fock calculation for a water
          cluster with 4158 water molecules.}
        \label{fig:water_SCF_acc_noacc_last_SCF_cycle}
\end{figure}

Figure \ref{fig:water_SCF_bounds_eigs} shows for each SCF cycle
intervals containing the homo and lumo eigenvalues propagated from the
previous SCF cycle (initial homo/lumo) and the improved intervals
computed using information extracted from the recursive expansion
(estimated homo/lumo).

\begin{figure}[ht!]
        \centering
	\includegraphics[width=\textwidth]{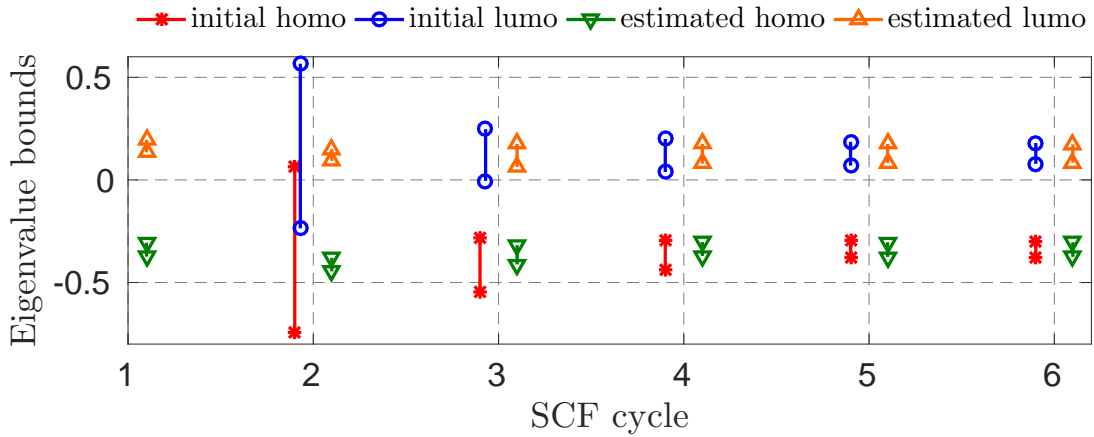}
        \caption{For each cycle in a self-consistent field Hartree--Fock
          calculation for a water cluster with 4158 water
          molecules with the SP2-ACC expansion,
          intervals containing the homo and lumo eigenvalues
          propagated from the previous SCF cycle (initial) and
          computed as a by-product of the recursive expansion
          (estimated). The corresponding figure for regular SP2 is
          essentially identical.
        \label{fig:water_SCF_bounds_eigs}
        }
\end{figure}

Table~\ref{tab:water_start_check_puri} shows upper and lower bounds
$n_{\textrm{max}}$ and $n_{\textrm{min}}$, respectively, and actual
number of iterations $n$ in the recursive expansion in each SCF
cycle. In the two initial SCF cycles the homo and lumo intervals are
overlapping, see Figure~\ref{fig:water_SCF_bounds_eigs}, and therefore
we cannot use Algorithm~\ref{alg:selection_polynomials_sp2acc} to
determine the sequence of polynomials. Thus, an upper bound of the
number of iterations $n_{\textrm{max}}$ based on the homo and lumo
intervals cannot be computed.  The acceleration has an effect as soon
as the outer bounds for homo and lumo are better than the extremal
bounds ($\lambda_\textrm{min}$ and $\lambda_\textrm{max}$), in our
case already in the second SCF cycle. However, in the initial SCF
cycles the outer bounds are loose, making the acceleration less
effective than in later iterations.  Starting from the third SCF cycle
the upper and lower bounds on the number of iterations are given by
Algorithm~\ref{alg:selection_polynomials_sp2acc}. In iteration
$n_{\textrm{min}}$ the acceleration has been turned off, as discussed
in Section~\ref{sec:sp2}.
Note that the proposed stopping criteria are used in each SCF cycle
independently of the method for choosing polynomials. The recursive
expansion is stopped as soon as the estimated order of convergence
computed using the mixed norm is smaller than $\tilde{q}=1.8$.

 \begin{table}[ht!] \centering
 \ra{1.1}
     \caption{Upper estimate $n_{\textrm{max}}$, actual number of
       iterations $n$ in recursive expansion in each SCF cycle and the
       iteration $n_{\textrm{min}}$ where the acceleration has been
       turned off. There is no acceleration in the first SCF cycle,
       and at least two iterations should be performed to be able to
       check the stopping criterion.}
   \begin{tabular}{l  l l l l l l} \toprule
   SCF cycle & 1 & 2 & 3 & 4 & 5 & 6\\
   \midrule
   $n_{\textrm{max}}$ & \,\texttt{-}  & \,\texttt{-}  & 26 & 24 & 24 & 23\\
   \hline
   $n$ & 29 & 20 & 19 & 17 & 18 & 17\\
   \hline
   $n_{\textrm{min}}$ & 2  & 14  & 14 & 14 & 14 & 14\\
   \bottomrule
    \end{tabular} 
    \label{tab:water_start_check_puri}
 \end{table}

The proposed stopping criteria, the acceleration technique and
efficient estimation of the homo and lumo eigenvalues give a
significant performance improvement of the recursive density matrix
expansion.  We show with the water cluster example that the use of the
SP2-ACC polynomials reduces the number of matrix-matrix
multiplications in comparison to the regular SP2 scheme and the
proposed stopping criteria enable us to reach the level of attainable
accuracy without spending redundant computational effort in the
stagnation phase.

The length of the conditioning phase depends on the homo-lumo gap.
Smaller homo-lumo gap will give a longer phase. The acceleration
technique reduces the length of the conditioning phase, but has no
significant impact in the purification phase. Our stopping criteria
are designed to automatically detect when numerical errors start to
dominate. By introducing a parameter $C_q$
satisfying~\eqref{eq:Cq_relation} we eliminate the need to determine
an iteration after which one can check the stopping criterion. The
proposed stopping criteria can be checked already in the first
iteration.  Their efficiency depends on the closeness of the chosen
value $C_q$ to the smallest value satisfying~\eqref{eq:Cq_relation},
but does not depend on the length of phases in the recursive
expansion, and is thus independent of the homo-lumo gap.

\section{Discussion}

The stopping criterion is an important aspect of the development
process of any iterative method. Many works address this question for
linear
systems~\cite{arioli1992stopping,arioli2013stopping,axelsson2001error,frommer2008stopping,kaasschieter1988practical}
and
eigensolvers~\cite{bennani1994stopping,Golub2000,vomel2008state}. In
general, such stopping criteria are based on controlling the norm of a
suitable residual or estimate of the norm of the error. The
convergence path can be irregular and contain stagnation regions.

Stopping criteria for iterative methods for matrix functions are often
based on the relative distance of the subsequent iterates $X_k$ and
$X_{k+1}$~\cite{book-higham}.
As soon as the distance becomes
smaller than some predefined threshold value the iterations stop.
However, in the presence of numerical errors it is hard to find an
optimal threshold value. 
To illustrate that our approach to develop stopping criteria for
iterative methods is applicable also to other matrix iterations, we
derive a stopping criterion of the same type for the Newton sign
matrix iteration in Appendix~\ref{sec:sign_matrix}.

Newton's method to find roots of real-valued functions is locally at
least quadratically convergent for simple roots. Also in this case,
the iterations are typically stopped when some error measure goes
below a predefined threshold value. Often one wants to continue
iterating until numerical errors (from e.g. floating point roundoff)
prevent any further decrease of the error.  The threshold value is
often chosen in terms of the expected accuracy of the evaluation of
the function, e.g. some multiple of machine epsilon. To avoid the
selection of threshold value one may, following the lines of the
present work, analyze the convergence behavior and devise a stopping
criterion based on the detection of a fall in convergence order, see
Appendix~\ref{sec:newton}.  This is in particular useful in cases when
the accuracy in the evaluation of a function is not known or depends
non-trivially on program input parameters.

\section{Concluding remarks}
Recursive expansions to compute the density matrix in electronic
structure calculations are usually stopped when the idempotency error
goes below some predefined tolerance. The main problem with such an
approach is that an appropriate value for the tolerance is difficult
to select. If the tolerance is small in relation to numerical errors
coming from removal of matrix elements or rounding errors, the
iterations will never stop. If the tolerance is large in relation to
numerical errors,
the same accuracy could be achieved with less effort.

In previous work we addressed the never stop issue by tightening the
tolerance for removal of small matrix elements at the end of the
expansion until the desired accuracy is achieved for
eigenvalues~\cite{m-accPuri}. However, this approach also involves a
user defined parameter for the accuracy in eigenvalues with potential
impact on convergence and computational cost.  

The practical usefulness of our new stopping criteria proposed here
can be seen in the context of the {\sc Ergo}~\cite{ergo_web,Ergo2011}
program where the new stopping criterion allows the number of input
parameters to the program to be reduced, since only a single parameter
for the density matrix construction accuracy is now needed. The new
stopping criterion also solves previous problems with failed
convergence when using default {\sc Ergo} parameters; previously, when using
the default parameters the recursive expansion failed due to rounding
errors in some cases, particularly for larger molecules where the
effect of rounding errors is more pronounced.

If the homo and lumo eigenvalues are known in advance one may compute
an upper bound for the number of iterations in advance by iterating
until the homo and lumo eigenvalues are within rounding error from
their desired values. The number of iterations in such an approach
corresponds to $n_{\textrm{max}}$ in the present work. For high
accuracy calculations this may be a reasonable approach but often it
would lead to superfluous iterations, as for example can be seen for
the water cluster calculations in Section~\ref{sec:scf_calc}.

In the present work three phases of the recursive expansion were
identified: conditioning, purification, and stagnation. The
appropriate moment to stop the expansion is at the transition between
purification and stagnation.  At this transition there is a drop in
the order of convergence.  By detection of this drop we are able to
stop the expansion at the appropriate moment without any user defined
parameters. The transition to stagnation is accurately detected even
if the idempotency error continues to slowly decrease.  By altering
the asymptotic error constant in the observed order of convergence we
avoid an early stop in the conditioning phase.

\begin{acknowledgement}

Support from the G{\"o}ran Gustafsson foundation, the Swedish research
council (grant no. 621-2012-3861), the Lisa and Carl--Gustav Esseen
foundation, and the Swedish national strategic e-science research
program (eSSENCE) is gratefully acknowledged. Computational resources
were provided by the Swedish National Infrastructure for Computing
(SNIC) at Uppsala Multidisciplinary Center for Advanced Computational
Science (UPPMAX).
  
\end{acknowledgement}

\begin{appendix}

\section{Sign matrix iterations}
\label{sec:sign_matrix}

Applying Newton's method to the function $f(X) = X^2-I$ gives an
iteration
\begin{align}
     X_{k+1} = \frac{1}{2}(X_k^{-1} + X_k), \quad X_0 = A,
\end{align}
for the sign matrix function which converges quadratically provided
that $A\in \mathbb{C}^{n\times n}$ has no eigenvalues on the imaginary
axis~\cite{book-higham}.

Since the matrices $X_k$ and $X_{k+1} - X_k$ commute, the Taylor
expansion of the matrix function $f(X_{k+1}) = X_{k+1}^2-I$ is given
by~\cite{Al-Mohy2009,Deadman2016}
\begin{align}
f(X_{k+1}) &= f(X_k) + f'(X_k)(X_{k+1} - X_k) + R_2(X_k) \\
&= f(X_k) + f'(X_k)(- (f'(X_k))^{-1}f(X_k)) + R_2(X_k) \\
&= R_2(X_k),
\end{align}
where the truncation error for the spectral norm~\cite{Mathias1993} is
bounded
\begin{align}
\|R_2(X_k)\|_2 &\leq \frac{1}{2} \|X_{k+1} - X_k\|_2^2 \max_{s\in[0,1]} \|f''(X_k + s(X_{k+1} - X_k))\|_2\\
&\leq \frac{1}{2}\|f(X_k)\|_2^2\|(f'(X_k))^{-1}\|_2^2  \max_{s\in[0,1]} \|f''(X_k + s(X_{k+1} - X_k))\|_2 \\
&=\frac{1}{4} \|X_k^{-1}\|_2^2 \|f(X_k)\|_2^2,
\end{align}
where we have used that $f'(X_k) = 2X_k$ and $ f''(X_k) = 2I$.

Thus in exact arithmetics
\begin{align}
\|f(X_{k+1})\|_2 \leq \frac{1}{4} \|X_k^{-1}\|_2^2 \|f(X_k)\|_2^2,
\end{align}
which suggests to stop the iterations as soon as
\begin{align}
 \|f(X_k)\|_2 < 1 \quad \text{and} \quad \frac{\log(\|f(X_{k+1})\|_2/C_k)}{\log(\|f(X_k)\|_2)} \leq 1.8,
\end{align}
where
\begin{align}
 C_k \geq \frac{1}{4} \|X_k^{-1}\|_2^2.
\end{align}
The value of $C_k$ should be chosen the smallest possible, see the
discussion in section~\ref{sec:paramless}.

If in addition we assume that the matrix $A$ is normal, then for
$k\geq 1$ the absolute values of the eigenvalues of $X_k$ are bounded
from below by 1 and $\|X_k^{-1}\|_2^2 \leq 1$.  Thus for all $k\geq 1$ we
define
\begin{align}
C_k := C = \frac{1}{4}.
\end{align}

\section{Stopping criteria for  Newton's method}
\label{sec:newton}

Let $f(x)$ be a twice continuously differentiable function, $f(x^*)=0$, and
$f'(x^*) \neq 0$. Then,  Newton's method
\begin{equation} \label{eq:newton}
  x_{k+1} = x_k + \Delta x_k, \quad \Delta x_k = -\frac{f(x_k)}{f'(x_k)}
\end{equation}
is locally quadratically convergent to $x^*$.

Assume that $f'(x_k) \neq 0$. Taylor expansion around $x_k$ with step
$\Delta x_k$ and using \eqref{eq:newton} and Lagrange's form of the
remainder gives
\begin{align}
f(x_{k+1}) & = \frac{f''(\xi_k)}{2(f'(x_{k}))^2}(f(x_{k}))^2
\end{align}
where $\xi_k$ is some value between $x_k$ and $x_{k+1}$.  Following
the ideas of the present article, this suggests the stopping
criterion: stop as soon as
\begin{align}
  |f(x_{k})| & < 1 \quad \textrm{ and } \quad 
\frac{\log |f(x_{k+1})/C_k|}{\log |f(x_{k})|} < 1.8,
\end{align}
where 
\begin{equation}
  C_k = \frac{F_k}{2(f'(x_{k}))^2}, \quad F_k \geq \max_{\xi \textrm{ between }x_{k} \textrm{ and } x_{k+1}} |f''(\xi)|.
\end{equation}
We note that if $f'(x_k)$ comes close to zero $C_k$ becomes very large
and the stopping criterion is not triggered. Thus, there is no need
for special treatment in such cases.

\end{appendix}

\bibliography{biblio}

\end{document}